\documentclass[twocolumn]{autart}    

\usepackage{amsfonts,amsmath,amssymb,mathrsfs}
\usepackage{epsfig,epstopdf,graphicx,graphics}
\usepackage{color}
\usepackage{float}
\usepackage{ifpdf}
\usepackage[colorlinks=true]{hyperref}
\usepackage{braket}
\usepackage{cite}
\usepackage{bbm}

\newcommand{\beq}{\begin{equation}}
	\newcommand{\eeq}{\end{equation}}
\newcommand{\beqm}{\begin{equation*}}
	\newcommand{\eeqm}{\end{equation*}}
\newcommand{\beqn}{\begin{eqnarray}}
	\newcommand{\eeqn}{\end{eqnarray}}
\newcommand{\beqnm}{\begin{eqnarray*}}
	\newcommand{\eeqnm}{\end{eqnarray*}}
\newcommand{\bea}{\begin{align}}
	\newcommand{\eea}{\end{align}}
\newcommand{\beam}{\begin{align*}}
	\newcommand{\eeam}{\end{align*}}
\newcommand{\bs}{\begin{subequations}}
	\newcommand{\es}{\end{subequations}}

\newcommand{\bei}{\begin{itemize}}
	\newcommand{\eei}{\end{itemize}}
\newcommand{\bed}{\begin{description}}
	\newcommand{\eed}{\end{description}}
\newcommand{\bee}{\begin{enumerate}}
	\newcommand{\eee}{\end{enumerate}}
\newcommand{\bey}{\begin{array}}
	\newcommand{\eey}{\end{array}}
\newcommand{\bef}{\begin{figure}}
	\newcommand{\eef}{\end{figure}}

\newcommand{\mbf}{\mathbf}

\def\l{\langle}
\def\r{\rangle}
\def\ff{\frac}

\newcommand{\la}{\label}

\renewcommand{\hbar}{\mathchar'26\mkern-9mu h}

\begin{document}

\begin{frontmatter}

	\title{Coherent feedback $H^\infty$ control of  quantum linear systems\thanksref{footnoteinfo}} 
	
	\thanks[footnoteinfo]{This paper was not presented at any IFAC meeting.  This work is partially financially supported by Quantum Science and Technology-National Science and Technology Major Project 2023ZD0300600, Guangdong Provincial Quantum Science Strategic Initiative No. GDZX2303007, Hong Kong Research Grant Council (RGC) under Grant No. 15213924, and  the CAS AMSS-PolyU Joint Laboratory of Applied Mathematics.
	}
	
	\author[polyu,polyuc]{Guofeng Zhang}\ead{guofeng.zhang@polyu.edu.hk},    
	\author[anu]{Ian~R.~Petersen}\ead{i.r.petersen@gmail.com}             
	
	\address[polyu]{Department of Applied Mathematics, The Hong Kong Polytechnic University, Hung Hom, Kowloon, Hong Kong SAR, China}
	\address[polyuc]{Research Institute for Quantum Technology, The Hong Kong Polytechnic University, Hong Kong SAR, China}      
	\address[anu]{School of Engineering, The Australian National University, Canberra ACT 2601, Australia}

	\begin{keyword}                          
		Quantum linear control systems; $H^\infty$ control; Lyapunov equations; algebraic Riccati equations.
	\end{keyword}

	\begin{abstract}                         
	The purpose of this paper is to investigate the coherent feedback $H^\infty$ control problem for linear quantum systems. A key contribution is a simplified design methodology that guarantees closed-loop stability and a prescribed level of disturbance attenuation. It is shown that for general linear quantum systems, a physically realizable quantum controller can be obtained by solving at most four Lyapunov equations. In the passive case, a necessary and sufficient condition is provided in terms of two uncoupled pairs of Lyapunov equations. These results represent a significant simplification over the standard approach, which requires solving two coupled algebraic Riccati equations. The effectiveness of the proposed method is demonstrated through two typical quantum optical devices:  an empty optical cavity and a degenerate parametric amplifier.  These results provide a computationally efficient procedure for the robust and optimal control of quantum optical and optomechanical systems.
	\end{abstract}
	
\end{frontmatter}

\section{Introduction}

In recent decades, substantial progress has been made in both the theoretical foundations and experimental realizations of control of quantum-mechanical systems. Quantum control plays a pivotal role in enabling key quantum technologies, including quantum communication, computation, cryptography, ultra-precision metrology, and nano-electronics. Within quantum control theory, quantum linear systems theory holds particular importance — much like linear systems do in classical control theory. Quantum linear systems serve as mathematical models for quantum harmonic oscillators, where the term ``linear” refers specifically to the linearity of the Heisenberg equations of motion for the system’s operators. This inherent linear structure considerably simplifies system analysis and controller synthesis, allowing powerful techniques from classical linear systems theory to be adapted effectively. As a result, quantum linear systems provide versatile and tractable models for a broad range of physical setups. Notable examples include quantum optical systems, circuit quantum electrodynamics (circuit QED), cavity QED systems, quantum optomechanical systems, atomic ensembles, and quantum memories  \cite{BK95,GZ00,WM08,WM10,NY17,DP23}.

To  the best of our  knowledge, coherent feedback $H^\infty$ control for linear quantum systems was first investigated in Ref. \cite{YK03b}, where simple examples  demonstrated how to use $H^\infty$ control to construct coherent feedback networks for achieving quantum squeezing. A general $H^\infty$ control framework for quantum linear systems was established in Ref. \cite{JNP08}. There, the authors derived a quantum version of the strict bounded real lemma as an extension of the classical one in Ref. \cite{PAJ91} for the study of the  nonsingular $H^\infty$ control problem and, based on it, developed a two-algebraic-Riccati-equation (ARE) solution for $H^\infty$ synthesis. A key contribution was the treatment of physical realizability — imposing quantum constraints such as the preservation of commutation relations.  It is fair to say that Ref. \cite{JNP08} underpins most subsequent studies on coherent feedback $H^\infty$ control of linear quantum systems. In Ref. \cite{MP11b}, coherent feedback $H^\infty$ control was studied for passive quantum linear systems whose dynamics are governed solely by annihilation operators, and the solution was given in terms of two coupled complex-domain AREs. The general (i.e., non-passive) case was later developed using the annihilation–creation operator representation in Ref. \cite{ZJ11}. The problem of time-varying coherent feedback $H^\infty$ control via a dynamic game approach was investigated in Ref. \cite{MP12}, with solutions expressed via two generalized Riccati differential equations and a spectral radius condition. Ref.~\cite{XPD17} addressed the robust   $H^\infty$ control of linear quantum systems by allowing uncertainties in the  system Hamiltonian, where the quantum controller is given in terms of two AREs. Time-varying quantum linear systems were further studied in Ref. \cite{LDP+22}, where coherent feedback  $H^\infty$ controllers were designed to achieve a prescribed disturbance attenuation level. Fluctuations in the amplitude and/or phase of the pump field of an optical parametric oscillator (OPO) can degrade its performance. Based on coherent feedback $H^\infty$ control, Ref.~\cite{LDP+22b} designs both passive and active quantum controllers to suppress such system fluctuations,  while simultaneously attaining the desired disturbance attenuation;  again, the main result is based on the solution of two coupled AREs. An equalization problem for passive linear quantum systems is studied in Refs. \cite{UJ24,XLDPU26} in the framework of coherent feedback $H^\infty$ control. Further research  can be found in Refs.~\cite{Mabuchi08, NJD09, ZLH+12, PDP17, CDZL17}, and the references therein.


In the study of coherent feedback control,  a  controller is designed to achieve the pre-specified control performance; in other words, the controller is constructed  by means of control theory. However, it may not be possible to physically realize such a controller  quantum-mechanically. This is the  problem of \emph{physical realizability of quantum-mechanical systems}.  Singular perturbation methods were  first used to demonstrate how to construct a quantum-mechanical controller in Ref. \cite{YK03b}.  A systematic procedure of constructing a quantum-mechanical controller from an $H^\infty$ controller is developed in Ref. \cite{JNP08}. Ref.  \cite{VP16} studies how to implement classical linear time-invariant (LTI) systems as physically realizable quantum systems by adding quantum noises. An algorithm is constructed to determine the minimal number of required auxiliary noise channels and presents a suboptimal coherent quantum LQG control algorithm.  An algorithm for physically implementable  passive quantum linear controllers is constructed in   Ref. \cite{UJ24}.  On the basis of physical realizability theory, Refs. \cite{nurdin2010synthesis,Nurdin2010b,petersen11}  studied how to build quantum controllers using quantum optical devices Refs. \cite{leonhardt2003,BR04}. More detail can be found in the book \cite{NY17}.

As mentioned above, most existing solutions to the coherent feedback $H^\infty$ control of quantum linear systems rely on the solution of two couples AREs, under certain assumptions; see for example  Assumption 5.2 in Ref. \cite{JNP08}, which are  identical to Assumptions \textbf{A1}-\textbf{A4} in Ref. \cite{PAJ91}.  In  Section \ref{sec:real}, it will be shown that for quantum linear systems,  Assumption \textbf{A1} naturally holds. Moreover, for the disturbance feedforward scenario  of $H^\infty$ control (see for example \cite[Section 16.6]{ZDG96}), Assumption \textbf{A2} naturally holds and Assumptions \textbf{A3} and \textbf{A4} are equivalent.
%
%
%
%
%
%
%
%
The main contribution of this paper is a significant simplification of coherent feedback \(H^{\infty}\) control design for quantum linear systems. The central result, Theorem \ref{thm_general}, shows that for a general linear quantum plant, a physically realizable linear quantum controller can be constructed by solving at most four Lyapunov equations. This replaces the conventional requirement to solve two coupled AREs in all the existing literature. Corollary \ref{cor:iff} provides a further simplification when the system $A$-matrix  is symmetric. For the important subclass of \emph{passive} quantum linear systems, Theorem \ref{thm:passive} establishes a simplified necessary and sufficient condition based on solving two pairs of Lyapunov equations. The practical application of the general theory is demonstrated through two detailed examples: an empty optical cavity and a degenerate parametric amplifier (DPA).

The remainder of this paper is organized as follows. Section \ref{sec:QLS} provides the preliminary background on quantum linear systems, presenting both the annihilation-creation and quadrature operator representations, along with their dynamical models and control-theoretic properties. In Section \ref{sec:real}, the coherent feedback \(H^{\infty}\) control problem for general quantum linear systems is formulated and investigated. The passive case is studied in Section \ref{sec:passive}. Applications of the general theory are then presented, with the $H^{\infty}$ control of an empty cavity examined in Section \ref{sec:cavity} and that of a DPA studied in Section \ref{sec:dpa}. Finally, Section \ref{sec:conclu} offers concluding remarks.

\textit{Notation}.
Let \(\imath = \sqrt{-1}\) denote the imaginary unit.  
Let $\mathbbm{R}$ be the field of real numbers, $\mathbbm{C}$ the field of complex numbers, and \(\mathbbm{Z}^+\) the set of positive integers.  
 For a column vector \(X = [x_1, \dots, x_n]^\top\) whose entries are complex numbers or operators, let \(X^\# = [x_1^\ast, \dots, x_n^\ast]^\top\) denote its complex conjugate or adjoint operator.  
Define \(X^\dagger = (X^\#)^\top\) and \(\breve{X} = [X^\top \; X^\dagger]^\top\).  
For two matrices \(U, V \in \mathbbm{C}^{k \times r}\), define the doubled‑up matrix   $\Delta(U,V)=\left[\begin{smallmatrix}
		U & V \\
		V^\# & U^\#
	\end{smallmatrix}\right]$.
 Denote by \(I_k\) the \(k \times k\) identity matrix.   
Let $J_k={\rm diag}\{I_k,-I_k\}$ and  $\mathbbm{J}_k=\bigl[
\begin{smallmatrix}
	0_{k} & I_k \\ 
	-I_k & 0_{k}%
\end{smallmatrix}
\big]$.
For \(X \in \mathbbm{C}^{2k \times 2r}\) define its \(\flat\)-adjoint and \(\sharp\)-adjoint respectively as  
$
X^\flat = J_r X^\dagger J_k$ and 
$X^\sharp = \mathbbm{J}_r^\top X^\dagger \mathbbm{J}_k$. 
The Kronecker delta is denoted by \(\delta_{jk}\), the Dirac delta function by \(\delta(t-r)\), and the tensor product by \(\otimes\).  
For operators \(a, b\), their commutator is \([a,b] = ab - ba\).  
Finally, the reduced Planck constant $\hbar$ is set to 1 in this paper.

%
%
%
%
%

\section{Quantum linear systems} \label{sec:QLS}

A quantum linear system $G$ models a collection of $n$ quantum harmonic oscillators driven by $m$ input Bosonic  fields. The $j$-th harmonic oscillator, $j = 1, \ldots, n$, is represented by its annihilation operator $\mbf{a}_{j}$ and creation operator $\mbf{a}_{j}^{\ast}$, which is the Hilbert space adjoint of $\mbf{a}_{j}$. These operators satisfy the canonical commutation relations (CCRs)
\beq \label{eq:CCR_a_2}
[\mbf{a}_{j}(t), \mbf{a}_{k}(t)] = 0, ~ [\mbf{a}_{j}^{\ast}(t), \mbf{a}_{k}^{\ast}(t)] = 0, ~ [\mbf{a}_{j}(t), \mbf{a}_{k}^{\ast}(t)] = \delta_{jk}.
\eeq
Let $\mbf{a} = [\mbf{a}_{1}~\cdots~\mbf{a}_{n}]^{\top}$. We can describe the quantum linear system in the $(S,\mbf{L},\mbf{H})$ formalism \cite{GJ09, GJ09b, ZJ12,ZLW+17,ZD22}. Here, the scattering operator $S\in\mathbbm{C}^{m\times m}$ is a unitary  matrix,  the system Hamiltonian is $\mbf{H} = \tfrac{1}{2}\mbf{\breve{a}}^{\dag}\Omega \mbf{\breve{a}}$, where $\Omega = \Delta(\Omega_{-}, \Omega_{+}) \in \mathbbm{C}^{2n \times 2n}$ is Hermitian with $\Omega_{-}, \Omega_{+} \in \mathbbm{C}^{n \times n}$, and the interface  between the system and the external fields is described by the operator $\mbf{L} = [C_{-}\; C_{+}]\,\mbf{\breve{a}}$, with $C_{-}, C_{+} \in \mathbbm{C}^{m \times n}$. The $k$-th input field, $k = 1, \ldots, m$, is modeled by its annihilation operator $\mbf{b}_{{\rm in}, k}(t)$ and creation operator $\mbf{b}_{{\rm in}, k}^{\ast}(t)$ (the Hilbert space adjoint of $\mbf{b}_{{\rm in},k}(t)$). These operators satisfy the singular commutation relations
\beq  \la{eq:CCR_b}
\begin{aligned}
& [\mbf{b}_{{\rm in}, j}(t), \mbf{b}_{{\rm in}, k}(r)] = 0, ~ [\mbf{b}_{{\rm in}, j}^{\ast}(t), \mbf{b}_{{\rm in}, k}^{\ast}(r)] = 0,  
\\
& [\mbf{b}_{{\rm in}, j}(t), \mbf{b}_{{\rm in}, k}^{\ast}(r)] = \delta_{jk}\delta(t-r).
\end{aligned}
\eeq
Let  $\mbf{b}_{\rm in}(t) = [\mbf{b}_{{\rm in}, 1}(t)~\cdots~\mbf{b}_{{\rm in}, m}(t)]^{\top}$. The Heisenberg equations of motion of the quantum linear system $G$  are governed by the It\^{o} quantum stochastic differential equations (QSDEs)
\beq\label{eq:sys}
\begin{aligned}
	\mbf{\dot{\breve{a}}}(t) 
	=&\; 
	\mathcal{A}\,\mbf{\breve{a}}(t) + \mathcal{B}\,\mbf{\breve{b}}_{\rm in}(t),  
	\\
	\mbf{\breve{b}}_{\mathrm{out}}(t) 
	=&\; 
	\mathcal{C}\,\mbf{\breve{a}}(t) + \mathcal{D}\,\mbf{\breve{b}}_{\rm in}(t), 
	\quad t \ge 0,
\end{aligned}
\eeq
with system matrices
\beq\label{ABCD0}
\begin{aligned}
	&\mathcal{D} = \Delta(S, 0),  \quad
	\mathcal{C} = \Delta(C_{-}, C_{+}), 
	\\
	& \mathcal{B} = -\mathcal{C}^{\flat}\mathcal{D},  \quad
	\mathcal{A} = -\imath J_{n}\Omega - \tfrac{1}{2}\mathcal{C}^{\flat}\mathcal{C}.
\end{aligned}
\eeq
The above system matrices satisfy physical realizability conditions  \cite{JNP08,NJP09,ZJ13}.
\begin{equation}\ \label{eq:PR_a_b} 
\mathcal{A}+\mathcal{A}^{\flat }+\mathcal{B}\mathcal{B}^{\flat } =0, \ \  
\mathcal{B} =-\mathcal{C}^{\flat }\mathcal{D}. 
\end{equation} 
The transfer matrix from $\breve{\mathbf{b}}_{\rm in}[s]$ to  $\breve{\mathbf{b}}_{\rm out}[s]$ is, \cite[Eq. (22)]{ZJ13},
\beq\label{eq:tf_G}
\Xi_G[s] = \mathcal{C} (sI-\mathcal{A})^{-1} \mathcal{B} + \mathcal{D} = \Delta(\Xi_{G_-}[s],\Xi_{G_+}[s] ),
\eeq
which satisfies $\Xi_G[-s^\ast]^\flat  \Xi_G[s]=  I_{2m},~~ \forall s \in \mathbbm{C}$.

The above discussions are in the Heisenberg picture, where the  joint system-field state remains constant.  Taking the quantum expectation on both sides of Eq. \eqref{eq:sys} with respect to the joint system-field initial state (\cite[Section 6.4.1]{WM10}, \cite[Section 2.6]{NY17}, \cite{NJP09}, \cite[Section 1]{BQD24}), yields a \emph{classical} linear system for mean dynamics \beq\la{eq:real_sys_mean}
\begin{aligned}
	\frac{d\l \breve{\mbf{a}}(t) \r}{dt} =&\; \mathcal{A} \l\breve{\mbf{a}}(t)\r + \mathcal{B} \l \breve{\mbf{b}}_{\rm in}(t)\r,
	\\
	\l\breve{\mbf{b}}_{\rm out}(t)\r =&\;  \mathcal{C} \l\breve{\mbf{a}}(t) \r+ \mathcal{D} \l \breve{\mbf{b}}_{\rm in}(t)\r.
\end{aligned}
\eeq
Thus, we can define  various control-theoretic notions for  the linear \emph{quantum} system \eqref{eq:sys} in terms of those for the linear \emph{classical} system \eqref{eq:real_sys_mean}.

\begin{defn}
	\label{def:stab_ctrb_obsv}  (\cite[Def. 1]{GZ15}, \cite[Def. 3.1]{ZD22} and \cite[Def. 2.1]{DZLP26})
	The linear quantum system \eqref{eq:sys} is said to be Hurwitz stable (resp. controllable, observable, detectable, stabilizable, minimal) if the corresponding linear classical system \eqref{eq:real_sys_mean} is Hurwitz stable (resp. controllable, observable, detectable, stabilizable, minimal).
\end{defn}

The justification of Definition \ref{def:stab_ctrb_obsv}  from a \emph{physical} point of view can be found in Ref. \cite{YK03b}.

If $C_+=0$ and $\Omega_+=0$,  the resulting
quantum linear system is said to be
\textit{passive} \cite{ZJ11,GZ15,NY17,ZGPG18}.  The It\^o QSDEs
for a passive linear
quantum system are
\beq\label{eq:sys_passive}
\begin{aligned}
	\dot{\mbf{a}} =&\; A \mbf{a} +B \mbf{b}_{\rm in},  \\
	\mbf{b}_{\rm out} =&\; C \mbf{a} +D  \mbf{b}_{\rm in},
\end{aligned}
\eeq
where 
\[
A=-\imath\Omega _{-}-\frac{1}{2}C_{-}^{\dagger }C_{-},  \ \  B=-C_{-}^{\dagger }S,\ C=C_{-}, \  D = S.
\]
Accordingly,  the physical realizability conditions~\eqref{eq:PR_a_b} reduce to
\begin{equation*}
	\label{eq:passive_PR}
	A+A^{\dagger }+BB^{\dagger }=0, ~B=-C^{\dagger }S.
\end{equation*}
Moreover, in the passive case,  $\Xi_{G^+}[s] \equiv 0 $ in Eq. \eqref{eq:tf_G}, and
\begin{equation}\label{eq:tf_passive_0}
	\Xi_{G^-}[s]  = S -C_-(sI+\imath\Omega _{-}+\frac{1}{2}C_{-}^{\dagger }C_{-})^{-1}C_-^\dagger S.
\end{equation}
Finally, it can be easily verified that for a quantum linear passive system, the following holds
\beq \label{eq:tf_passive}
\Xi_{G^-}[\imath\omega]^\dagger  \Xi_{G^-}[\imath\omega] \equiv I_m, ~~~ \forall \omega\in \mathbbm{R}.
\eeq
As a result, a quantum linear passive system is like an all-pass filter that does not change the amplitude of the input signal, but modifies its phase.   According to \cite[Lemma 2]{GZ15},  controllability,  observability and Hurwitz stability are equivalent to each other. Consequently, if  the passive linear system \eqref{eq:sys_passive} is a minimal realization, then by \cite[Theorem 2]{P13}, it is physically realizable if and only the associated transfer matrix  in \eqref{eq:tf_passive_0} satisfies Eq. \eqref{eq:tf_passive}.

Example 2.1 in Ref. \cite{DZLP26} demonstrated that detectability and stabilizability may not be equivalent for general linear quantum systems. However, they are equivalent in the passive case. 

\begin{prop} \label{prop}
The passive linear system \eqref{eq:sys_passive} is stabilizable if and only if it is detectable.
\end{prop}

The proof of Proposition \ref{prop} is straightforward, thus is omitted.


So far we have used the complex annihilation–creation operator representation. An alternative is the real quadrature operator representation \cite[Sec.~II.E]{ZJ11}.   For a positive integer $k$,  define the unitary matrix
\beqm \label{eq:unitary_V}
V_{k}= \frac{1}{\sqrt{2}}\left[\begin{array}{@{}cc@{}}I_{k} & I_{k} \\
	-\imath I_{k} & \imath I_{k}
\end{array}
\right].
\eeqm
The following coordinates transformations
\beq  \label{complex_to_real_trans}
\begin{aligned}
&	\left[\begin{array}{c}\mbf{q} \\
		\mbf{p}
	\end{array}
	\right] \equiv  \mbf{x} = V_{n}\breve{\mbf{a}}, 
	\left[\begin{array}{c}
		\mbf{q}_{\mathrm{in}} \\
		\mbf{p}_{\mathrm{in}}
	\end{array}
	\right] \equiv \mbf{u} =V_{m}\breve{\mbf{b}}_{\rm in},
	\\
	&\left[\begin{array}{c}\mbf{q}_{\mathrm{out}} \\
		\mbf{p}_{\mathrm{out}}
	\end{array}
	\right] \equiv \mbf{y} = V_{m}\breve{\mbf{b}}_{
		\mathrm{out}}
\end{aligned}
\eeq
generate real quadrature operators of the system and the fields. The counterparts of the  commutation relations \eqref{eq:CCR_a_2} and \eqref{eq:CCR_b} are $[\mbf{x}(t), \mbf{x}(t)^\top] =\imath \mathbbm{J}_n$, 
and
$[\mbf{u}(t), \ \mbf{u}^\top(r)] = \imath\delta(t-r)\mathbbm{J}_m, \ t,r\in\mathbbm{R}$
respectively. In terms of unitary transformations in Eq. \eqref{complex_to_real_trans}, the coupling operator $\mbf{L}$ and the Hamiltonian $\mbf{H}$ are transformed to $\mbf{L} =  \Lambda\mbf{x}$, and  $\mbf{H}= \frac{1}{2}\mbf{x} ^\top\mathbbm{H} \mbf{x}$, where $\Lambda   = [C_{-}\ C_{+}]V_n^\dagger$ and  $\mathbbm{H}  = V_n \Omega V_n^\dagger$. Therefore, the
QSDEs that describe the dynamics of the linear quantum system  in the real quadrature operator representation are
\beq\la{eq:real_sys}
\begin{aligned}
\dot{\mbf{x}} =&\; \mathbbm{A} \mbf{x} + \mathbbm{B} \mbf{u},
\\
\mbf{y} =&\; \mathbbm{C} \mbf{x} + \mathbbm{D} \mbf{u},
\end{aligned}
\eeq
where
\beq \label{eq:real_sys_ABCD}
\begin{aligned}
\mathbbm{D} =&\; V_{m} \mathcal{D} V_{m}^{\dagger} =\left[
\bey{@{}cc@{}}
\mathrm{Re}(S) & -\mathrm{Im}(S) \\
\mathrm{Im}(S)  & \mathrm{Re}(S)
\eey
\right],
\\
\mathbbm{C} =&\; V_{m} \mathcal{C}
V_{n}^{\dagger}=
\left[
\bey{@{}ll@{}}
\mathrm{Re}(C_-+C_+) & -\mathrm{Im}(C_--C_+)\\
\mathrm{Im}(C_-+C_+) & \mathrm{Re}(C_--C_+)
\eey
\right],
\\
\mathbbm{B} =&\; V_{n} \mathcal{B} V_{m}^{\dagger} 
=-\left[
\bey{@{}ll@{}}
\mathrm{Re}(C_-^\dag-C_+^\dag) & -\mathrm{Im}(C_-^\dag-C_+^\dag)\\
\mathrm{Im}(C_-^\dag+C_+^\dag) & \mathrm{Re}(C_-^\dag+C_+^\dag)
\eey
\right]\mathbbm{D} ,
\\
\mathbbm{H} =&\;V_{n} \Omega V_{n}^{\dagger} = \left[
\begin{array}{cc}
	\mathrm{Re}(\Omega_-+\Omega_+)  & -\mathrm{Im}(\Omega_--\Omega_+)   \\
	\mathrm{Im}(\Omega_-+\Omega_+)  & \mathrm{Re}(\Omega_--\Omega_+)
\end{array}
\right],
\\
\mathbbm{A} =&\; V_{n} \mathcal{A} V_{n}^{\dagger} = \mathbbm{J}_n \mathbbm{H} -\frac{1}{2}\mathbbm{C}^{\sharp}\mathbbm{C}.
\end{aligned}
\eeq
It can be readily verified that the matrix $\mathbbm{J}_n \mathbbm{H}$ is a Hamiltonian matrix \cite[Sec. 7.8]{GvL13}, and hence we have
\beq \label{eq:J_nH}
( \mathbbm{J}_n\mathbbm{H})^\sharp = -\mathbbm{J}_n\mathbbm{H}.
\eeq
In the real-quadrature-operator representation, the physical realizability
conditions (\ref{eq:PR_a_b}) take the form
\begin{equation} \la{eq:PR_real}
	\mathbbm{A} + \mathbbm{A}^{\sharp} + \mathbbm{B}\mathbbm{B}^{\sharp}=0,~\mathbbm{B}=-\mathbbm{C}^{\sharp}\mathbbm{D}^{\sharp}.
\end{equation}
 We denote the transfer matrix of system \eqref{eq:real_sys} by $\mathbbm{G}(s)$. Clearly,
\beq\label{eq:G_real}
\mathbbm{G}(s) = V_m G(s)V_m^\dag,
\eeq
which satisfies
\beq \la{eq:jun7_temp2}
\mathbbm{G}(-s^\ast)^\sharp \mathbbm{G}(s)=
\mathbbm{G}(s)\mathbbm{G}(-s^\ast)^\sharp= I_{2m},
\eeq
see, e.g.,  \cite{GJN10,ZJ13,GZ15,ZD22} for more details. 

In the \emph{passive} case, the system matrices in  Eq. \eqref{eq:real_sys_ABCD} become
\beq \label{eq:real_sys_ABCD_passive}
\begin{aligned}
\mathbbm{D} =&\left[
\bey{@{}cc@{}}
\mathrm{Re}(S) & -\mathrm{Im}(S) \\
\mathrm{Im}(S)  & \mathrm{Re}(S)
\eey
\right], ~~ 
\mathbbm{C} =
\left[
\bey{@{}ll@{}}
\mathrm{Re}(C_-) & -\mathrm{Im}(C_-)\\
\mathrm{Im}(C_-) & \mathrm{Re}(C_-)
\eey
\right],
\\
\mathbbm{B} =&-\mathbbm{C}^\top\mathbbm{D} , ~~
\mathbbm{H} = \left[
\begin{array}{cc}
	\mathrm{Re}(\Omega_-)  & -\mathrm{Im}(\Omega_-)   \\
	\mathrm{Im}(\Omega_-)  & \mathrm{Re}(\Omega_-)
\end{array}
\right],
\\
\mathbbm{A} =&\;  \mathbbm{J}_n \mathbbm{H} -\frac{1}{2}\mathbbm{C}^\top \mathbbm{C}.
\end{aligned}
\eeq

We end this section with a result for passive linear quantum systems.

\begin{prop}
The passive quantum linear system \eqref{eq:real_sys} is Hurwitz stable if and only if and only if the matrix $\mathbbm{A}$ is non-singular.  The matrix $A$ is singular if there are less input fields than the number of the system harmonic oscillators.
\end{prop}

\emph{Proof.}  Applying rotations (unitary operations)
\beq \label{eq:rotation}
\begin{aligned}
&\mbf{a}(t) \to e^{\imath  \Omega_-  t}\mbf{a}(t), ~ \mbf{b}_{\rm in}(t)  \to   e^{\imath \Omega_- t } \mbf{b}_{\rm in}(t),
\\
&  \mbf{b}_{\rm out}(t)  \to   e^{\imath \Omega_-  t} \mbf{b}_{\rm out}(t)
\end{aligned}
\eeq
transform the passive linear quantum system \eqref{eq:sys_passive} to
\beq\label{eq:sys_passive_rotation}
\begin{aligned}
\dot{\mbf{a}} =&\; -\frac1{2}C^\dag C \mbf{a} -C^\dag S \mbf{b}_{\rm in},  \\
\mbf{b}_{\rm out} =&\;  C\mbf{a} +S \mbf{b}_{\rm in}.
\end{aligned}
\eeq
As a result, under the rotations in Eq. \eqref{eq:rotation},   the term  $\mathbbm{J}_n\mathbbm{H}$  vanishes and therefore $\mathbbm{A}  =  -\frac{1}{2}\mathbbm{C}^\top\mathbbm{C}$, which is negative semi-definite. Thus, the quantum linear system \eqref{eq:real_sys} is Hurwitz stable if and only if $\mathbbm{A}$ is non-singular. If there are less input fields that the number of the system harmonic oscillators, then the matrix $\mathbbm{C}\in \mathbbm{R}^{2m\times 2n}$ cannot have full column rank and thus the matrix $\mathbbm{A}$ must be singular.  \hfill $\Box$

\section{Coherent feedback  $H^\infty$ control: the general case}\label{sec:real}

In this section, we investigate the coherent feedback $H^\infty$ control problem for general quantum linear systems within the framework of the real-quadrature operator representation. The coherent feedback network  is shown in Fig. \ref{fig_cl}. 

\begin{figure}[h!]
	\centering
	\includegraphics[width=0.4\textwidth]{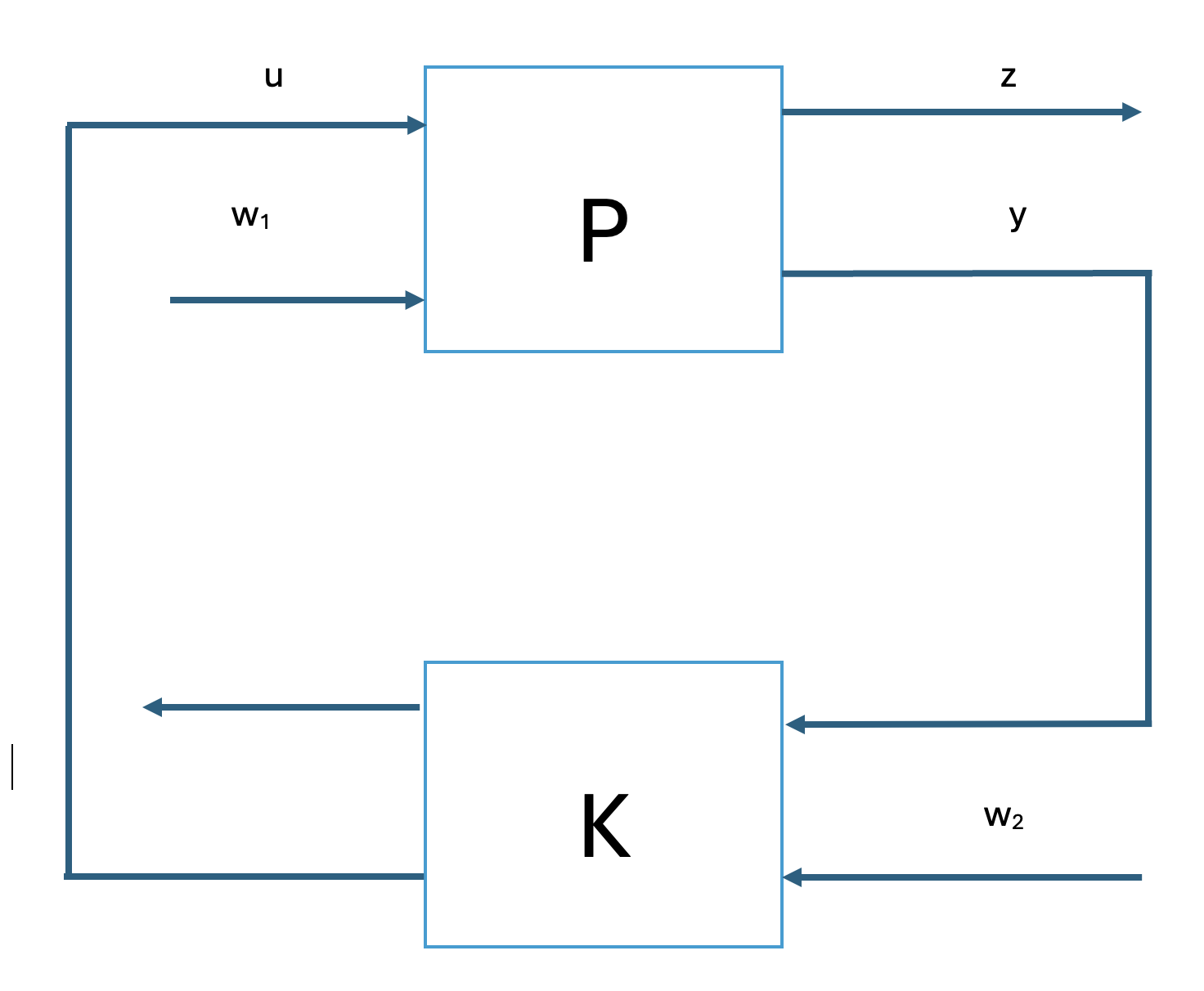}\\
	\caption{The coherent feedback system composed of a quantum plant $P$ and a quantum controller $K$. }
	\label{fig_cl}
\end{figure}

The plant $P$ in the real-quadrature operator representation  is
\beq \label{eq: plant_0}
\begin{aligned}
	\dot{\mbf{x}} =&\; \mathbbm{A} \mbf{x} +  \mathbbm{B}_1 \mbf{w}_1 + \mathbbm{B}_2 \mbf{u}, \\
	\mbf{z}           =&\; \mathbbm{C}_1 \mbf{x} + \mathbbm{D}_{12} \mbf{u},\\
	\mbf{y}           =&\; \mathbbm{C}_2\mbf{x} +  \mathbbm{D}_{21}  \mbf{w}_1.
\end{aligned}
\eeq
Here, $\mbf{u}$ is a vector of $k$ control input fields, and   $\mbf{w}_1$ is a vector of $l$ disturbance input fields of the form
 \beq \label{eq:signal+noise}
 \mbf{w}_1 (t)  = \beta_u(t) +\overset{\sim}{\mbf{w}}_1 (t),
 \eeq
 where $\beta_u(t)$ are  self-adjoint, adapted processes which encode the signal part of the quantum field $\mbf{w}_1$, while $\overset{\sim}{\mbf{w}}_1 (t)$ is quantum vacuum noise \cite{HP84,YK03b,JNP08}. An example of $ \mbf{w}_1$ is a laser beam which is  undesirable (in other words, a disturbance; see for examples Refs.  \cite{JNP08,Mabuchi08}).  Another example is that $ \mbf{w}_1$ is the output of another quantum system \cite[Fig. 1]{ZJ11}.  $\mbf{z}$ is the output quantum field corresponding to the input field $\mbf{u}$ and $\mbf{y}$ is the output quantum field corresponding to the input field  $\mbf{w}_1$.   According to Eq. \eqref{eq:real_sys_ABCD}, the system matrices are 
\beq \label{eq:ABCD_0}
\begin{aligned}
\mathbbm{A} =&\; \mathbbm{J}_n \mathbbm{H} -\frac{1}{2}\mathbbm{C}_1^{\sharp}\mathbbm{C}_1-\frac{1}{2}\mathbbm{C}_2^{\sharp}\mathbbm{C}_2, \\
\mathbbm{B}_1 =&\; -\mathbbm{C}_2^\sharp \mathbbm{D}_{21},\\
 \mathbbm{B}_2 =&\; -\mathbbm{C}_1^\sharp \mathbbm{D}_{12}.
\end{aligned}
\eeq
As discussed in Section \ref{sec:QLS},  in the $(S,\boldsymbol{L},\boldsymbol{H})$ formalism,  $S$ is a unitary matrix. 
As a result, by Eqs. \eqref{ABCD0} and \eqref{eq:real_sys_ABCD},  both $\mathbbm{D}_{12}$  and $\mathbbm{D}_{21}$ are real orthogonal matrices.

Let
\beq \label{eq:E1E2}
E_1 = \mathbbm{D}_{12}^\top \mathbbm{D}_{12}= I, \quad E_2 =\frac1{\gamma^2} \mathbbm{D}_{21}\mathbbm{D}_{21}^\top= \frac1{\gamma^2}I .
\eeq

Clearly, Assumptions \textbf{A1} and \textbf{A2} in  Ref. \cite{PAJ91}, which are items 1) and 2) in Assumption 5.2 in Ref. \cite{JNP08}, hold naturally.

The purpose is to design a quantum controller $K$ of the form 
\beq \label{eq:controller_real}
\begin{aligned}
	\dot{\mbf{x}}_K =&\; \mathbbm{A}_K \mbf{x}_K +\mathbbm{B}_K \mbf{y} + \tilde{\mathbbm{B}}_K \mbf{w}_2,\\
	\mbf{u} =&\; \mathbbm{C}_K \mbf{x}_K + \mbf{w}_2,\\
	\tilde{\mbf{u}} =&\; \tilde{\mathbbm{C}}_K \mbf{x}_K +\mbf{y},
\end{aligned}
\eeq
so that the feedback system  in Fig. \ref{fig_cl}  is internally stable and the $H^\infty$ norm of the closed-loop transfer matrix $T_{\mbf{w}_1\to \mbf{z}} $ is less than a prescribed disturbance attenuation level  $\gamma$. The output channel $\tilde{\mbf{u}}$ is not shown in Fig. \ref{fig_cl} as it is not used.   As in Ref. \cite{JNP08},
matrices $\mathbbm{A}_K,\mathbbm{B}_K,\mathbbm{C}_K$ are to be determined by the $H^\infty$ control design. On the other hand,  the physical realizability of the quantum controller $K$ demands
\beq\label{eq:dec16_BK_CK}
\tilde{B}_K =  -\mathbbm{C}_K^\sharp, \quad   \tilde{\mathbbm{C}}_K = -\mathbbm{B}_K^\sharp.
\eeq
The input field $\mbf{w}_2$  to the quantum controller \eqref{eq:controller_real} is  in the vacuum state. Thus, averaging over the joint initial system-field state yields
\beq\label{eq:controller_real_mean}
\begin{aligned}
\langle	\dot{\mbf{x}}_K\rangle =&\; \mathbbm{A}_K\langle	 \mbf{x}_K \rangle+\mathbbm{B}_K \langle	\mbf{y}\rangle,\\
	\langle	\mbf{u}\rangle =&\; \mathbbm{C}_K \langle\mbf{x}_K\rangle,\\
\langle		\tilde{\mbf{u}}\rangle =&\;  \tilde{\mathbbm{C}}_K\langle	 \mbf{x}_K \rangle+\langle	\mbf{y}\rangle,
\end{aligned}
\eeq
which is the form of the \emph{central solution} in \cite[Eq. (4)]{PAJ91} if the unused $\langle\tilde{\mbf{u}}\rangle $ is ignored; see Fig. \ref{fig_cl}. 

%
%
%

Define matrices
\beq\label{eq:A_x_A_y_dec13}
\begin{aligned}
\mathbbm{A}_x =&\; \mathbbm{A}-\mathbbm{B}_2 E_1^{-1} \mathbbm{D}_{12}^\top  \mathbbm{C}_1 = \mathbbm{J}_n\mathbbm{H}+\ff1{2}\mathbbm{C}_1^\sharp \mathbbm{C}_1 -\ff1{2}\mathbbm{C}_2^\sharp \mathbbm{C}_2, \\
 \mathbbm{A}_y  =&\; \mathbbm{A}-\ff1{\gamma^2}\mathbbm{B}_1 \mathbbm{D}_{21}^\top  E_2^{-1} \mathbbm{C}_2 = \mathbbm{J}_n\mathbbm{H}-\ff1{2}\mathbbm{C}_1^\sharp \mathbbm{C}_1 +\ff1{2}\mathbbm{C}_2^\sharp \mathbbm{C}_2. 
\end{aligned}
\eeq
By Eq. \eqref{eq:J_nH},  we have
\beq \label{eq:oct27_A_x_y}
 \mathbbm{A}_y  = - \mathbbm{A}_x^\sharp.
\eeq

\begin{rem} \label{rem:D12_D21}
{\rm 
Notice that the quantum plant \eqref{eq: plant_0} is slightly different from that in \cite[Eq. (21)]{JNP08}. 
\bei
\item First, in our setting the matrix $\mathbbm{D}_{12}$ is real orthogonal. According to Assumption \textbf{A1} in Ref. \cite{PAJ91}, which is item 1) in \cite[Assumption 5.2]{JNP08}, the matrix  $D_{12}$ in Ref. \cite{JNP08}   is of full column rank. Consequently, in Ref. \cite{JNP08}   the output vector $z$ has at least the same dimension as the input vector $u$. To maintain the commutation relations of the input and output fields, $D_{12}$ must therefore be square—and hence real orthogonal. As a result, our $\mathbbm{D}_{12}$ is the same as the $D_{12}$ in \cite[Eq. (21)]{JNP08}.

\item Second, \cite[Eq. (21)]{JNP08} includes an additional matrix $D_{20}$, which corresponds  to the quantum vacuum noise $v(t)$ there. If $v$ is treated as a disturbance with zero signal part; see Eq. \eqref{eq:signal+noise}, it can be merged with $w$, and accordingly $D_{20}$ can be grouped together with $D_{21}$. From this viewpoint, \cite[Eq. (21)]{JNP08} takes the same form as Eq. \eqref{eq: plant_0} in this paper. Thus, for simplicity, we assume the absence of the quantum noise $v(t)$ in our system.

\item  Item 2) of Assumption 5.2 in \cite{JNP08} requires the matrix $D_{21}$ be of full row rank.  This implies that the dimension of the input vector $w$ is at least that of the output vector $y$, meaning that some output channels associated with  the inputs $w$ may not be used for feedback. Combining $v$ and $w$ would further increase the number of such unused output channels. 
In contrast, we assume that the matrix  $\mathbbm{D}_{21}$ is real orthogonal, so the number of outputs $\mbf{y}$ equals the number of inputs $\mbf{w}_1$ and all channels participate in the feedback.  This is the disturbance feedforward scenario  of $H^\infty$ control; see for example  \cite[Section 16.6]{ZDG96}.

\item Finally, as established above, the matrix \(D_{12}\) in \cite[Eq. (21)]{JNP08} is real orthogonal, and \(D_{21}\) in \cite[Eq. (21)]{JNP08}  has full row rank. Consequently, Eq. \eqref{eq:A_x_A_y_dec13} holds.
Hence,  Eq. \eqref{eq:oct27_A_x_y}  is inherent to the structure of the type of linear quantum systems studied in this paper. 
\eei
}
\end{rem}

Assumptions \textbf{A3} and \textbf{A4} in Ref. \cite{PAJ91}, which are  items  3) and 4) in \cite[Assumption 5.2]{JNP08}, are respectively
\begin{subequations} 	\label{eq:A3A4}
	\beqn
	&& {\rm rank} \left[
	\bey{cc}
	\mathbbm{A}-\imath \omega I & \mathbbm{B}_2 \\
	\mathbbm{C}_1   & \mathbbm{D}_{12}
	\eey
	\right]  = 2n+2k  \ \ {\rm for ~all} \ \ \omega\geq 0,
	\label{eq:A3}
	\\
	&& {\rm rank} \left[
	\bey{cc}
	\mathbbm{A}-\imath \omega I & \mathbbm{B}_1 \\
	\mathbbm{C}_2   & \mathbbm{D}_{21}
	\eey
	\right]  = 2n+2l  \ \ {\rm for ~all} \ \ \omega\geq 0.
	\label{eq:A4}
	\eeqn
\end{subequations}

\begin{rem}{\rm
It can be verified that under the condition that  the matrix $\mathbbm{D}_{12}$ is real orthogonal, Eq. \eqref{eq:A3} is equivalent to that the matrix $\mathbbm{A}_x$ has no purely imaginary eigenvalues. Similarly,  under the condition that  the matrix $\mathbbm{D}_{21}$ is real orthogonal,  Eq. \eqref{eq:A4} is equivalent to that the matrix $\mathbbm{A}_y$ has no purely imaginary eigenvalues. Moreover,  by Eq. \eqref{eq:oct27_A_x_y}, $\lambda$ is an eigenvalue of $\mathbbm{A}_x$ if and only if $-\lambda$ is an eigenvalue of $\mathbbm{A}_y$.   As a result, Assumption \textbf{A3} is equivalent to Assumption \textbf{A4}.  This is a special feature of \emph{quantum} linear systems.
}
\end{rem}

Similar to \cite[Eqs. (2)-(3))]{PAJ91} or \cite[Eqs. (28)-(29)]{JNP08}, we consider the following AREs:
\begin{subequations} \label{eq:ARE_0}
	\beqn
	\mathbbm{A}_x^\top \mathbbm{X} + \mathbbm{X} \mathbbm{A}_x + \mathbbm{X}(\frac1{\gamma^2}  \mathbbm{B}_1 \mathbbm{B}_1^\top -\mathbbm{B}_2 \mathbbm{B}_2^\top)\mathbbm{X} &=&0 , \label{eq:X_0}\\
	\mathbbm{A}_y \mathbbm{Y} + \mathbbm{Y}  \mathbbm{A}_y^\top + \mathbbm{Y} (\mathbbm{C}_1^\top \mathbbm{C}_1 - \gamma^2 \mathbbm{C}_2^\top \mathbbm{C}_2)\mathbbm{Y}  &=& 0.  \label{eq:Y_0}
	\eeqn
\end{subequations}
Note that the last terms on the left-hand side of \cite[Eqs. (2)-(3))]{PAJ91} or \cite[Eqs. (28)-(29)]{JNP08} vanish due to  Eqs. \eqref{eq:ABCD_0}-\eqref{eq:E1E2}.

The following lemma is a direct consequence of \cite[Theorem 3.1]{PAJ91}; see also \cite[Theorem 5.4]{JNP08}.

\begin{lem}\label{lem1}
Assume  Eq. \eqref{eq:A3A4} holds.   There exists a quantum controller $K$ such that the closed-loop feedback system in  Fig. \ref{fig_cl} is internally stable and achieves disturbance attenuation  $T_{\mbf{w}_1\to \mbf{z}} <\gamma$ if and only if the AREs  \eqref{eq:ARE_0} have solutions   $\mathbbm{X} ,\mathbbm{Y} >0$ such that  the spectral radius $\rho(\mathbbm{X}\mathbbm{Y})<1$, and both $\mathbbm{A}_x+(\frac1{\gamma^2}  \mathbbm{B}_1 \mathbbm{B}_1^\top -\mathbbm{B}_2 \mathbbm{B}_2^\top)\mathbbm{X} $ and $\mathbbm{A}_y+\mathbbm{Y}(\mathbbm{C}_1^\top \mathbbm{C}_1 - \gamma^2 \mathbbm{C}_2^\top \mathbbm{C}_2) $ are Hurwitz stable.
\end{lem}

Applying Schur decomposition to $\mathbbm{A}_x$ in Eq. \eqref{eq:A_x_A_y_dec13}  yields
\beq \label{eq:tilde_A}
\mathbbm{W} \mathbbm{A}_x \mathbbm{W}^\top = \left[
\bey{cc}
\mathbbm{A}_{x1} & \mathbbm{A}_{x2}\\
0   & \mathbbm{A}_{x3} 
\eey
\right] = \tilde{\mathbbm{A}}_x,
\eeq
where $\mathbbm{W}$ is  orthogonal,   $ \mathbbm{A}_{x1}$ is stable and $ \mathbbm{A}_{x3}$ is anti-stable (i.e, $-\mathbbm{A}_{x3}$ is stable). Accordingly, denote 
\beq\label{eq:tilde_B}
\mathbbm{W}  \mathbbm{B}_1 = \left[
\bey{c}
\mathbbm{B}_{1x1}\\
\mathbbm{B}_{1x2}
\eey
\right] = \tilde{\mathbbm{B}}_1, \quad 
\mathbbm{W} \mathbbm{B}_2= \left[
\bey{c}
\mathbbm{B}_{2x1}\\
\mathbbm{B}_{2x2}
\eey
\right]= \tilde{\mathbbm{B}}_2.
\eeq

The following result shows that instead of solving the coupled AREs \eqref{eq:X_0}-\eqref{eq:Y_0}, it suffices to solve at most four Lyapunov equations.

\begin{thm}\label{thm_general}
Assume  Eq. \eqref{eq:A3} holds.   There exists a quantum controller $K$ such that the closed-loop feedback system in  Fig. \ref{fig_cl} is internally stable and achieves disturbance attenuation  $T_{\mbf{w}_1\to \mbf{z}} <\gamma$ if  there are $\mathbbm{S},\mathbbm{T},\mathbbm{U},\mathbbm{V}>0$ which solve the Lyapunov equations  
\begin{subequations}\label{eq:S and T_real}
	\beqn
	-\mathbbm{A}_{x3}\mathbbm{S} - \mathbbm{S} \mathbbm{A}_{x3}^\top +\mathbbm{B}_{2x2} \mathbbm{B}_{2x2}^\top &=&0, \\
	-\mathbbm{A}_{x3} \mathbbm{T} - \mathbbm{T} \mathbbm{A}_{x3}^\top  +\mathbbm{B}_{1x2}\mathbbm{B}_{1x2}^\top &=& 0 ,
	\eeqn
\end{subequations}
and 
\beq\label{eq:U and V_real}
\begin{aligned}
\mathbbm{A}_{x1}\mathbbm{U} + \mathbbm{U} \mathbbm{A}_{x1}^\top +\mathbbm{B}_{1x1}\mathbbm{B}_{1x1}^\top =&\; 0, \\
\mathbbm{A}_{x1} \mathbbm{V} + \mathbbm{V} \mathbbm{A}_{x1}^\top  +\mathbbm{B}_{2x1}\mathbbm{B}_{2x1}^\top =&\; 0 ,
\end{aligned}
\eeq
subject to  $S-
\frac{T}{\gamma^2} >0$, $U-
\frac{V}{\gamma^2} >0$, $ \mathbbm{A}_{x2}\left(U-
\frac{V}{\gamma^2}\right)  + \left(U-
\frac{V}{\gamma^2}\right) \mathbbm{A}_{x2}^\top =0  $, and $\sigma_{\rm min} \left(  \mathbbm{S}-
\frac{\mathbbm{T}}{\gamma^2}
\right)  \sigma_{\rm min} \left(\mathbbm{U}-
\frac{\mathbbm{V}}{\gamma^2} \right) < \gamma^2 $.
\end{thm}

\emph{Proof.} The proof follows the procedure in Refs. \cite{P89,PAJ91}. Notice that Eq. \eqref{eq:X_0} is of the same form as the equation (ARE) in \cite{P89}, while Eq. \eqref{eq:Y_0} is not.  However, the replacement  $\hat{\mathbbm{Y} } =\gamma^2   \mathbbm{J} ^\top \mathbbm{Y}   \mathbbm{J}$ converts Eq. \eqref{eq:Y_0} to
\beq\label{eq:Y_tilde_real}
-\mathbbm{A}_x^\top  \hat{\mathbbm{Y} }  - \hat{\mathbbm{Y} }\mathbbm{A}_x +\hat{\mathbbm{Y} } (\frac1{\gamma^2}  \mathbbm{B}_2 \mathbbm{B}_2^\top -\mathbbm{B}_1 \mathbbm{B}_1^\top)\hat{\mathbbm{Y} } =0.
\eeq
which is of the same form as Eq. \eqref{eq:X_0}. 
%
Denote
\beq\label{eq:tilde_XY}
\tilde{\mathbbm{X}} = \mathbbm{W}\mathbbm{X}\mathbbm{W}^\top, ~  \tilde{\mathbbm{Y}} = \mathbbm{W}\hat{\mathbbm{Y}} \mathbbm{W}^\top =\gamma^2   \mathbbm{W}\mathbbm{J} ^\top \mathbbm{Y}   \mathbbm{J} \mathbbm{W}^\top,
\eeq
where the matrix $\mathbbm{W}$ is that in Eq. \eqref{eq:tilde_A}. Then the problem of finding $\mathbbm{X}$ and $\mathbbm{Y}$ becomes finding $\tilde{\mathbbm{X}}$ and $\tilde{\mathbbm{Y}}$. We look at Eq. \eqref{eq:X_0} first. For the anti-stable matrix $\mathbbm{A}_{x3}$, there exist $\mathbbm{S},\mathbbm{T}\geq0$ which solve the Lyapunov equations \eqref{eq:S and T_real}. Assume $S-
\frac{T}{\gamma^2} >0$ and define
\beq \label{eq:X_dec17l}
\tilde{\mathbbm{X}}  = \left[
\bey{cc}
0 & 0 \\
0 & \left(\mathbbm{S}-
\frac{\mathbbm{T}}{\gamma^2}
\right)^{-1}
\eey
\right].
\eeq
Clearly, the resulting $\mathbbm{X} = W^\top\tilde{\mathbbm{X}} W^\top $ solves the ARE \eqref{eq:X_0}. Next, we look at Eq. \eqref{eq:Y_0}. By Eqs. \eqref{eq:tilde_A}-\eqref{eq:tilde_XY},  Eq. \eqref{eq:Y_tilde_real} becomes
\beq\label{eq:Y_tilde_real_dec18} 
-\tilde{\mathbbm{A}}_x^\top \tilde{\mathbbm{Y}} -\tilde{\mathbbm{Y}} \tilde{\mathbbm{A}}_x+ \tilde{\mathbbm{Y} } (\frac1{\gamma^2}  \tilde{\mathbbm{B}}_2 \tilde{\mathbbm{B}}_2^\top -\tilde{\mathbbm{B}}_1 \tilde{\mathbbm{B}}_1^\top)\tilde{\mathbbm{Y} }=0.
\eeq
Assume  $\tilde{\mathbbm{Y}} $ is of the form
\beq
\tilde{\mathbbm{Y}}  = \left[
\bey{cc}
\tilde{\mathbbm{Y}} _1 & 0\\
0 &0
\eey
\right].
\eeq
Then Eq. \eqref{eq:Y_tilde_real_dec18}   is equivalent to \eqref{eq:dec17:11}-\eqref{eq:dec17:12} 
\begin{subequations}
\beqn
&& -\mathbbm{A}_{x1}^\top \tilde{\mathbbm{Y}} _1  -\tilde{\mathbbm{Y}} _1  \mathbbm{A}_{x1}
\nonumber
\\
&& \quad \quad  \quad +\tilde{\mathbbm{Y}} _1  \left( \frac1{\gamma^2} \mathbbm{B}_{2x1} \mathbbm{B}_{2x1}^\top -\mathbbm{B}_{1x1} \mathbbm{B}_{1x1}^\top  \right)=0, 
\label{eq:dec17:11}\\
&&\mathbbm{A}_{x2}^\top \tilde{\mathbbm{Y}} _1  + \tilde{\mathbbm{Y}} _1  \mathbbm{A}_{x2} = 0.
\label{eq:dec17:12}
\eeqn
\end{subequations}
For the anti-stable matrix $-\mathbbm{A}_{x1}$,   there exist $\mathbbm{U},\mathbbm{V}\geq0$ which solve the  Lyapunov equations \eqref{eq:U and V_real}.  If $U-
\frac{V}{\gamma^2} >0$, then $\tilde{\mathbbm{Y}}_1 =\left(\mathbbm{U}-
\frac{\mathbbm{V}}{\gamma^2}
\right)^{-1} $ solves  Eq. \eqref{eq:dec17:11}.  Moreover, if $ \mathbbm{A}_{x2}\left(U-
\frac{V}{\gamma^2}\right)  + \left(U-
\frac{V}{\gamma^2}\right) \mathbbm{A}_{x2}^\top =0  $, then $\tilde{\mathbbm{Y}}_1$ satisfies Eq. \eqref{eq:dec17:12}.


In the following, we look at $\rho(\mathbbm{X}\mathbbm{Y})<1$ in Lemma \ref{lem1}. By Eq. \eqref{eq:tilde_XY},  
\beq  \label{eq:XY_dec23}
\begin{aligned}
 \mathbbm{X}\mathbbm{Y} =& \frac1{\gamma^2} \mathbbm{W}^\top \tilde{\mathbbm{X}}\mathbbm{W} \mathbbm{J}  \mathbbm{W}^\top \tilde{\mathbbm{Y}}  \mathbbm{W} \mathbbm{J}^\top
\\
=&   \frac1{\gamma^2} \mathbbm{W}^\top \tilde{\mathbbm{X}}\mathbbm{W} \mathbbm{J}  \mathbbm{W}^\top \mathbbm{J}^\top \mathbbm{J} \tilde{\mathbbm{Y}}\mathbbm{J}^\top \mathbbm{J}  \mathbbm{W} \mathbbm{J}^\top.
\end{aligned}
\eeq
Noticing
\beqm
\tilde{\mathbbm{X}} =\left[
\bey{cc}
0 & 0 \\
0 & \left(\mathbbm{S}-
\frac{\mathbbm{T}}{\gamma^2}
\right)^{-1}
\eey
\right], \quad 
 \mathbbm{J} \tilde{\mathbbm{Y}}\mathbbm{J}^\top
=
\left[
\bey{cc}
0& 0 \\
0 &   \left(\mathbbm{U}-
\frac{\mathbbm{V}}{\gamma^2}
\right)^{-1} 
\eey 
\right],
\eeqm
we have
\beqm
\begin{aligned}
\mathbbm{X}\mathbbm{Y} =&\; \frac1{\gamma^2} \mathbbm{W}^\top \left[
\bey{cc}
0 & 0 \\
0 & \left(\mathbbm{S}-
\frac{\mathbbm{T}}{\gamma^2}
\right)^{-1}
\eey
\right]  \mathbbm{W} \mathbbm{J}  \mathbbm{W}^\top \mathbbm{J}^\top 
\\
& \times  \left[
\bey{cc}
0& 0 \\
0 &   \left(\mathbbm{U}-
\frac{\mathbbm{V}}{\gamma^2}
\right)^{-1} 
\eey 
\right]  \mathbbm{J}  \mathbbm{W} \mathbbm{J}^\top .
\end{aligned}
\eeqm
Define a matrix 
\beq \label{eq:Z_dec23}
\mathbbm{Z} =  \mathbbm{J}  \mathbbm{W} \mathbbm{J}^\top \mathbbm{W}^\top=    (\mathbbm{W}^\top)^\sharp \mathbbm{W}^\top,
\eeq
which is  skew-Hamiltonian and real orthogonal. 
Then
\beq \label{eq:rho_XYZ}
\begin{aligned}
&\rho(\mathbbm{X}\mathbbm{Y} ) 
\\
=&\; \frac1{\gamma^2} \rho\left(  \mathbbm{Z} \left[
\bey{cc}
0 & 0 \\
0 & \left(\mathbbm{S}-
\frac{\mathbbm{T}}{\gamma^2}
\right)^{-1}
\eey
\right] \mathbbm{Z}^\top  \left[
\bey{cc}
0& 0 \\
0 &   \left(\mathbbm{U}-
\frac{\mathbbm{V}}{\gamma^2}
\right)^{-1} 
\eey 
\right] \right).
\end{aligned}
\eeq
As $\mathbbm{Z}$ is a skew-Hamiltonian matrix \cite[Sec. 7.8]{GvL13}, it has the form 
\beq \label{eq:Z_dec23b}
\mathbbm{Z} = \left[ 
\bey{cc}
\mathbbm{Z}_1 & \mathbbm{Z}_2 \\
\mathbbm{Z}_3  & \mathbbm{Z}_1^\top
\eey
\right],
\eeq
where $\mathbbm{Z}_2$ and $\mathbbm{Z}_3$ are skew-symmetric.  By means of Eq. \eqref{eq:Z_dec23b} we can rewrite Eq. \eqref{eq:rho_XYZ} as
\beqm
\rho(\mathbbm{X}\mathbbm{Y} ) = \frac1{\gamma^2} \rho\left( 
\left[
\bey{cc}
0 & \mathbbm{Z}_2   \left(\mathbbm{S}-
\frac{\mathbbm{T}}{\gamma^2}
\right)^{-1} \mathbbm{Z}_1  \left(\mathbbm{U}-
\frac{\mathbbm{V}}{\gamma^2}
\right)^{-1}\\
0  & \mathbbm{Z}_1^\top  \left(\mathbbm{S}-
\frac{\mathbbm{T}}{\gamma^2}
\right)^{-1} \mathbbm{Z}_1  \left(\mathbbm{U}-
\frac{\mathbbm{V}}{\gamma^2}
\right)^{-1}
\eey
\right]
 \right),
\eeqm
which means
\beq \label{eq:XY_dec23_temp2}
\begin{aligned}
& \rho(\mathbbm{X}\mathbbm{Y} ) <1  
\\
 \Longleftrightarrow &
\rho\left(\mathbbm{Z}_1^\top  \left(\mathbbm{S}-
\frac{\mathbbm{T}}{\gamma^2}
\right)^{-1} \mathbbm{Z}_1  \frac1{\gamma^2} \left(\mathbbm{U}-
\frac{\mathbbm{V}}{\gamma^2}
\right)^{-1}
 \right)<1 .
 \end{aligned}
\eeq
If $ \left(\mathbbm{S}-
\frac{\mathbbm{T}}{\gamma^2}
\right)^{-1}>0$ and  $\left(\mathbbm{U}-
\frac{\mathbbm{V}}{\gamma^2}
\right)^{-1}>0$,   Eq. \eqref{eq:XY_dec23_temp2} can be rewritten as 
\begin{align}
&\rho(\mathbbm{X}\mathbbm{Y} ) <1  
\label{eq:XY_dec23_temp3}
\\
 \Longleftrightarrow& 
\rho\left(\mathbbm{F}^\top  \left(\mathbbm{S}-
\frac{\mathbbm{T}}{\gamma^2}
\right)^{-1} \mathbbm{F} \frac1{\gamma^2} \left(\mathbbm{U}-
\frac{\mathbbm{V}}{\gamma^2}
\right)^{-1}
\right)<1,
\nonumber
\end{align}
where $ \left(\mathbbm{S}-
\frac{\mathbbm{T}}{\gamma^2}
\right)^{-1}$ and  $\left(\mathbbm{U}-
\frac{\mathbbm{V}}{\gamma^2}
\right)^{-1}$ are transformed to  their SVD form, and  the resulting matrix $\mathbbm{F}$ satisfies $\mathbbm{F}^\top \mathbbm{F}\leq I$, and  $\mathbbm{F} \mathbbm{F}^\top\leq I$.  Clearly, $\mathbbm{F}$ satisfying $\mathbbm{F}^\top \mathbbm{F}\leq I$ and  $\mathbbm{F} \mathbbm{F}^\top\leq I$ is related to $\mathbbm{Z}_1$ and thus is  not arbitrary. If we let  $\mathbbm{F}$ be an arbitrary matrix that satisfies $\mathbbm{F}^\top \mathbbm{F}\leq I$ and  $\mathbbm{F} \mathbbm{F}^\top\leq I$, then  Eq. \eqref{eq:XY_dec23_temp3}  is equivalent to
\begin{align}
&\rho(\mathbbm{X}\mathbbm{Y} ) <1
\la{eq:jan2_iff}
\\
   \Longleftrightarrow  &
\sigma_{\rm max} \left(  \left(\mathbbm{S}-
\frac{\mathbbm{T}}{\gamma^2}
\right)^{-1} \right)  \sigma_{\rm max} \left( \left(\mathbbm{U}-
\frac{\mathbbm{V}}{\gamma^2}
\right)^{-1} \right) < \gamma^2.
\nonumber
\end{align}
Consequently, a sufficient condition is 
\begin{align}
&\sigma_{\rm max} \left(  \left(\mathbbm{S}-
\frac{\mathbbm{T}}{\gamma^2}
\right)^{-1} \right)  \sigma_{\rm max} \left( \left(\mathbbm{U}-
\frac{\mathbbm{V}}{\gamma^2}
\right)^{-1} \right) < \gamma^2
\nonumber
\\
\Longrightarrow&
\rho (\mathbbm{X}\mathbbm{Y})<1 ,
\label{eq:suff_XY}
\end{align}
which is equivalent to 
\beq \label{eq:suff_XY_2}
\sigma_{\rm min} \left(  \mathbbm{S}-
\frac{\mathbbm{T}}{\gamma^2}
 \right)  \sigma_{\rm min} \left(\mathbbm{U}-
\frac{\mathbbm{V}}{\gamma^2} \right) < \gamma^2 
\Longrightarrow
\rho (\mathbbm{X}\mathbbm{Y})<1 .
\eeq
(Notice that the LHS of Eq. \eqref{eq:suff_XY_2} is the condition in the statement of the theorem.) \hfill $\Box$

\begin{rem}{\rm
Two AREs are needed in the standard approach, as given in Refs. \cite{PAJ91,JNP08}. Here, we need only to solve at most four Lyapunov equations which are linear equations.  Moreover,   if $A_x$ is  either stable or anti-stable, solving two Lyapunov equations is sufficient. An example is given in Subsection \ref{subsec:case 1}.
}
\end{rem}

If both $\Omega_-$ and $\Omega_+$ are purely imaginary while both $C_-$ and $C_+$ are real, or if all four parameters --- $\Omega_-$, $\Omega_+$, $C_-$, and $C_+$--- are purely imaginary, it is straightforward to verify that the matrix $\mathbbm{A}$ in Eq. \eqref{eq:real_sys_ABCD} is symmetric. This property arises in many quantum linear systems, such as empty cavities (Section \ref{sec:cavity} below), degenerate parametric amplifiers (Section \ref{sec:dpa} below), rotated non-degenerate parametric amplifiers \cite[Section 1.5.2]{NY17}, Ref. \cite{SY21}, and the model studied in Ref. \cite{BQND24} (if the perturbation term is neglected). For such systems,  the sufficient condition in Theorem \ref{thm_general} can be strengthened to the following necessary and sufficient condition.

\begin{cor}\label{cor:iff}
Assume  Eq. \eqref{eq:A3} holds.   If the matrix $\mathbbm{A}_x$ is symmetric and  the coordinates transformation matrix $\mathbbm{Z}$ in Eq. \eqref{eq:Z_dec23} is an identity matrix, then there exists a quantum controller $K$ such that the closed-loop feedback system in  Fig. \ref{fig_cl} is internally stable and achieves disturbance attenuation  $T_{\mbf{w}_1\to \mbf{z}} <\gamma$ if and only if   there are matrices $\mathbbm{S},\mathbbm{T},\mathbbm{U},\mathbbm{V}>0$ which solve the Lyapunov equations   \eqref{eq:S and T_real} and \eqref{eq:U and V_real},  and satisfy  $\mathbbm{S}-
\frac{\mathbbm{T}}{\gamma^2} >0$, $\mathbbm{U}-
\frac{\mathbbm{V}}{\gamma^2} >0$,   and $\sigma_{\rm min} \left(  \mathbbm{S}-
\frac{\mathbbm{T}}{\gamma^2}
\right)  \sigma_{\rm min} \left(\mathbbm{U}-
\frac{\mathbbm{V}}{\gamma^2} \right) < \gamma^2 $.
\end{cor}

\textbf{Proof.} If the matrix $\mathbbm{A}_x$ is symmetric, it is diagonalizable. In this case, $\mathbbm{A}_{x2}=0$ and therefore  Eq. \eqref{eq:dec17:12} vanishes. On the other hand,  if the coordinates transformation matrix $\mathbbm{Z}$ in Eq. \eqref{eq:Z_dec23} is an identity matrix, then  Eq. \eqref{eq:jan2_iff} holds, and Eq. \eqref{eq:suff_XY_2} becomes 
\beq \label{eq:suff_XY_jan2}
\sigma_{\rm min} \left(  \mathbbm{S}-
\frac{\mathbbm{T}}{\gamma^2}
\right)  \sigma_{\rm min} \left(\mathbbm{U}-
\frac{\mathbbm{V}}{\gamma^2} \right) < \gamma^2 
\Longleftrightarrow
\rho (\mathbbm{X}\mathbbm{Y})<1 .
\eeq
The result follows.  $\hfill$ $\Box$

Corollary \ref{cor:iff} applies to the passive case to be studied in Section \ref{sec:passive} and the  DPA  example in Subsection \ref{subsec:case 1}.

After obtaining $\mathbbm{X}$ and $\mathbbm{Y}$ in Eq. \eqref{eq:XY_dec23} that solve the AREs \eqref{eq:ARE_0},  we can construct the quantum controller $K$.  Specifically,  like  \cite[Eq. (5)]{PAJ91} or  \cite[Eq. (30)]{JNP08}, the system matrices for the controller $K$ in Eq. \eqref{eq:controller_real} are
\beq\label{eq:controller_ABC_real}
\begin{aligned}
	\mathbbm{A}_K =&\;  \mathbbm{A}+\mathbbm{B}_2\mathbbm{C}_K -\mathbbm{B}_K \mathbbm{C}_2 + 
	\frac1{\gamma^2}(\mathbbm{B}_1-\mathbbm{B}_K \mathbbm{D}_{21}) \mathbbm{B}_1^\top \mathbbm{X} , \\
	\mathbbm{B}_K  =&\;  (I-\mathbbm{Y}\mathbbm{X})^{-1} (\gamma^2 \mathbbm{Y}\mathbbm{C}_2^\top +\mathbbm{B}_1\mathbbm{D}_{21}^\top),\\
	\mathbbm{C}_K =&\;  -(\mathbbm{B}_2^\top \mathbbm{X} +\mathbbm{D}_{12}^\top \mathbbm{C}_1),\\	
	\tilde{B}_K =&\; -C_K^\sharp, \quad \tilde{C}_K = -B_K^\sharp.
	\end{aligned}
\eeq

\begin{rem}{\rm
To ensure the physical realizability  of the controller $K$ in Eq. \eqref{eq:controller_real},   additional vacuum input fields may have to be added to make the controller genuinely quantum-mechanical.   However, this will not affect $\|T_{\mbf{w}_1\to \mbf{z}} \|_\infty $.  The interested reader may refer to Refs. \cite{JNP08,NJD09,VP16,UJ24} for more details. 
}
\end{rem}

\section{Coherent feedback  $H^\infty$ control: the passive case}\label{sec:passive}

In  this section, we assume both the plant $P$  and to-be-designed controller $K$   in Fig. \ref{fig_cl} are passive linear quantum systems.  For the passive case, under the rotations in Eq. \eqref{eq:rotation}, they are both of the form \eqref{eq:sys_passive_rotation}.

To be consistent with  notations used in Fig. \ref{fig_cl}, in this section we still adopt the  symbols in the real-quadrature representation. However, they  actually represent their counterparts in the  annihilation-creation operator representation. For example,  $\mbf{x}$ is actually $\mbf{a}$,   $\mbf{w}_1$ and $\mbf{u}$  are actually $\mbf{b}_{\rm in,1}$ and $\mbf{b}_{\rm in,2}$, respectively.

%

Similar to the development in Section \ref{sec:real}, consider the  plant $P$ in the annihilation-creation operator representation
\beq\label{eq: plant P}
\begin{aligned}
	\dot{\mbf{x}} =&\; A\mbf{x} + \frac1{\gamma}B_1 \mbf{w}_1 + B_2 \mbf{u}, \\
	\mbf{z}           =&\; C_1 \mbf{x} + D_{12} \mbf{u},\\
	\mbf{y}           =&\; C_2\mbf{x} +  \frac1{\gamma} D_{21}  \mbf{w}_1,
\end{aligned}
\eeq
where the constant system matrices are
\beqm
\begin{aligned}
A =&\;  -\frac1{2}C_1^\dag C_1 -\frac1{2}C_2^\dag C_2,\\
B_1 =&\;  -C_2^\dag D_{21},  ~~ B_2 = -C_1^\dag D_{12},
\end{aligned}
\eeqm
with $D_{12}$ and  $D_{21}$ being unitary. Thus, 
$E_1 = D_{12}^\dag D_{12}= I$ and $E_2 =\frac1{\gamma^2} D_{21}D_{21}^\dag= \frac1{\gamma^2}I$. We have
\beq \label{eq:AxAy_dec16}
\begin{aligned}
		A_x =&\; A-B_2 E_{1}^{-1}D_{12}^\dag C_1   =\frac1{2}C_1^\dag C_1 -\frac1{2}C_2^\dag C_2,\\
	A_y = & \;A-B_1 D_{21}^\dag E_2^{-1}C_2 =  -\frac1{2}C_1^\dag C_1 +\frac1{2}C_2^\dag C_2.
\end{aligned}
\eeq
Clearly,
\beq \label{eq:A_x_y_complex}
A_y = A_y^\dag = -A_x = -A_x^\dag.
\eeq
As a result,  Eq. \eqref{eq:A3A4} holds if and only if the matrix $A_x$ is nonsingular. 

Similar to Eq. \eqref{eq:ARE_0}, we consider the following AREs:
\begin{subequations} \label{eq:ARE}
	\beqn
	A_x^\dag X + XA_x + X(\frac1{\gamma^2}  B_1B_1^\dag -B_2B_2^\dag)X &=&0 , \label{eq:X}\\
	A_y Y + Y A_y^\dag + Y(C_1^\dag C_1 - \gamma^2 C_2^\dag C_2)Y &=& 0.  \label{eq:Y}
	\eeqn
\end{subequations}


%

Notice that Eq. \eqref{eq:X} is of the same form as the equation (ARE) in Ref.  \cite{P89}, while Eq. \eqref{eq:Y} is not. 
Let $\tilde{Y} = \gamma^2 Y$.  By means of Eq.\eqref{eq:A_x_y_complex}, Eq. \eqref{eq:Y} is converted to 
\beq\label{eq:Y_tilde_2}
	-A_x^\dag  \tilde{Y}  - \tilde{Y}A_x +\tilde{Y} (\frac1{\gamma^2}  B_2 B_2^\dag -B_1 B_1^\dag)\tilde{Y} =0,
\eeq
which is of the same form as Eq. \eqref{eq:X}.

Assume that the Hermitian matrix $A_x$  is nonsingular. By Schur decomposition we get
\beqm \label{eq:decomp}
A_x= \left[
\bey{cc}
A_{x1} & 0\\
0          & A_{x3}
\eey
\right], \  \ B_2 =  \left[
\bey{c}
B_{2x1}\\
 B_{2x2}
\eey
\right], \  \ B_1 =  \left[
\bey{c}
B_{1x1}\\
B_{1x2}
\eey
\right],
\eeqm
where $A_{x1}$ is stable and $A_{x3}$ is anti-stable.  For the anti-stable $A_{x3}$, there exist $S,T\geq0$ which solve the Lyapunov equations:
\beq\label{eq:S and T}
\begin{aligned}
-A_{x3} S - S A_{x3}^\dag +B_{2x2}B_{2x2}^\dag =&\; 0, \\
-A_{x3} T - T A_{x3}^\dag +B_{1x2}B_{1x2}^\dag =&\; 0.
\end{aligned}
\eeq
On the other hand, as $A_y = -A_x^\dag$, for the anti-stable  matrix $-A_{x1}$, there exist $U,V\geq0$ that solve the Lyapunov equations:
\beq\label{eq:U and V}
\begin{aligned}
	A_{x1} U + U \mathbbm{A}_{x1}^\dag +B_{1x1}B_{1x1}^\dag =&\; 0, \\
	A_{x1}^\dag V + V \mathbbm{A}_{x1}^\dag +B_{2x1}B_{2x1}^\dag =&\; 0.
\end{aligned}
\eeq
Define matrices
\beq \label{eq:XY}
X = \left[
\bey{cc}
0 & 0 \\
0 & \left(S-
\frac{T}{\gamma^2}
\right)^{-1}
\eey
\right], \ \ 
 \tilde{Y}  
 = 
\left[
\bey{cc}
\left(U-
\frac{V}{\gamma^2}
\right)^{-1}
& 0 \\
0 &  0
\eey
\right]. 
\eeq
Clearly,
\beq \label{eq:XY_complex}
\rho(XY)  = \frac1{\gamma^2}\rho (X\tilde{Y})=0.
\eeq

\begin{rem}{\rm 
In the passive case, Eq. \eqref{eq:dec17:12} vanishes. Moreover, due to the absence of $\mathbbm{J}$, the matrix $\mathbbm{Z}$ in Eq. \eqref{eq:Z_dec23} reduces to an identity matrix.
}
\end{rem} 

We are ready to present the main result of this section.
\begin{thm}\label{thm:passive}
Assume the  matrix $A_x$ in Eq. \eqref{eq:AxAy_dec16} is nonsingular.  There exists a quantum controller $K$ such that the closed-loop feedback system in  Fig. \ref{fig_cl} is internally stable and  achieves disturbance attenuation  $T_{\mbf{w}_1\to \mbf{z}} <\gamma$ if and only if there are $S,T,U,V>0$ which respectively solve the Lyapunov equations   \eqref{eq:S and T} and \eqref{eq:U and V}, and satisfy $S-
\frac{T}{\gamma^2} >0$ and $U-
\frac{V}{\gamma^2} >0$.
\end{thm}

%


\section{Examples} \label{sec:example}

In this section, we use two typical devices in quantum optics: an empty cavity and a DPA, to illustrate the theories proposed in Sections \ref{sec:real} an \ref{sec:passive}.

\subsection{Example 1: empty cavity} \label{sec:cavity}

\begin{figure}[h!]
	\centering
	\includegraphics[width=0.4\textwidth]{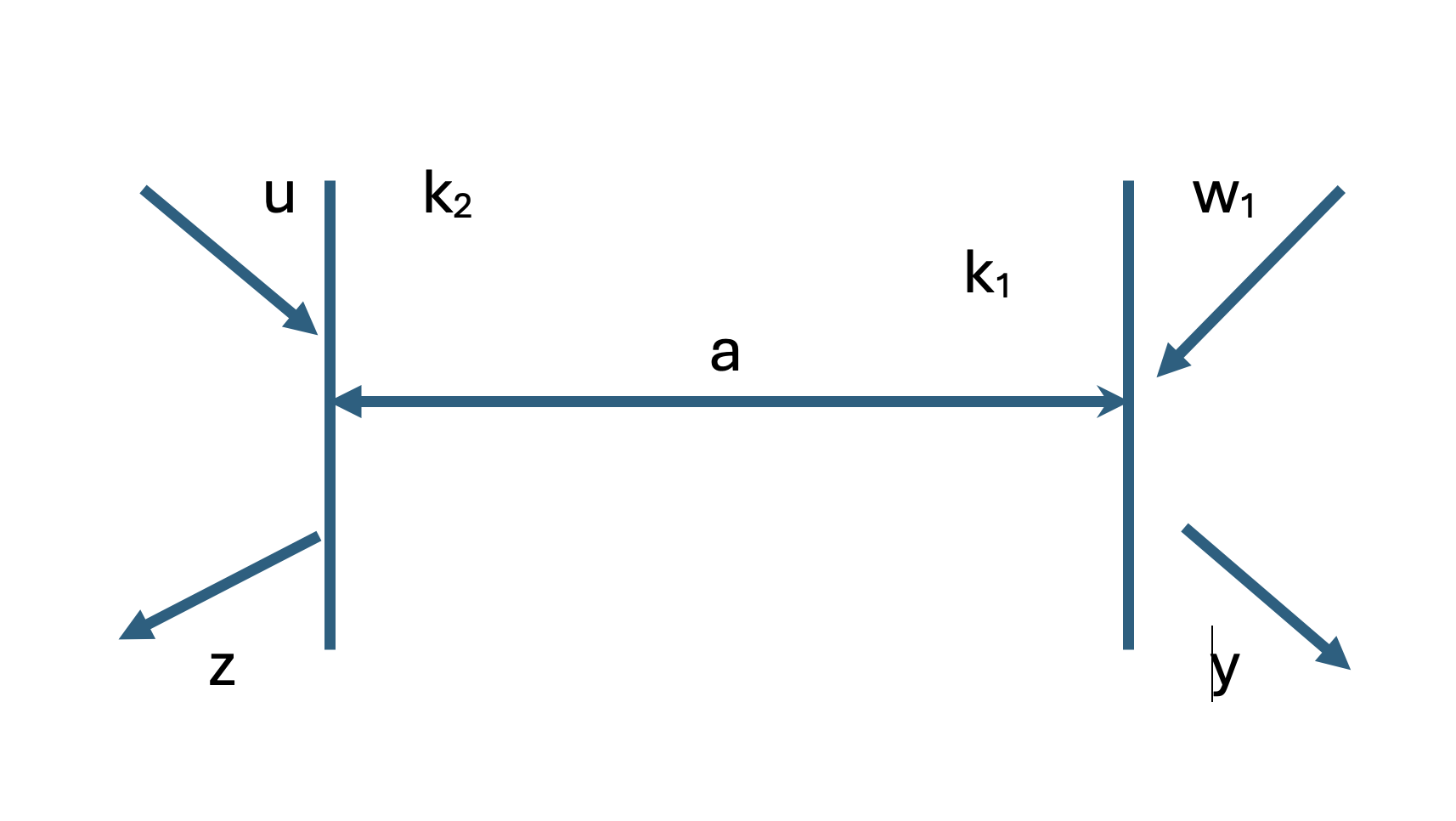}
	\caption{An empty cavity}
	\label{fig_cavity}
\end{figure}

Consider a single-mode cavity with two inputs, as shown in Fig. \ref{fig_cavity}, of the form
\beq \la{eq:K}
\begin{aligned}
	\dot{\mbf{a}} =&\; 
	- \frac{\kappa_1+\kappa_2}{2} \mbf{a}
	- \frac1{\gamma} \sqrt{\kappa_1} \mbf{w}_1 -\sqrt{\kappa_2} \mbf{u}, 
	\\
	\mbf{z}  =&\; 
	\sqrt{\kappa_2} \mbf{a} +  \mbf{u},
	\\
	\mbf{y} =&\; \sqrt{\kappa_1} \mbf{a}+ \frac1{\gamma} \mbf{w}_1.
	\end{aligned}
\eeq
By Eq. \eqref{eq:AxAy_dec16} we have $
A_x =  \frac{-\kappa_1+\kappa_2}{2}$ and $
A_y = \frac{\kappa_1-\kappa_2}{2} = -A_x$,
which confirm Eq. \eqref{eq:A_x_y_complex}.  In what follows, let $\kappa_2>\kappa_1$. Then $A_x>0$ and $A_y<0$.  By Eq. \eqref{eq:S and T}, we have 
$S = \frac{\kappa_2}{\kappa_2-\kappa_1}>0$ and $T=  \frac{\kappa_1}{\kappa_2-\kappa_1}>0$. 
Hence,
\beq \label{eq:gamma_cavity}
S-\frac{T}{\gamma^2}>0   \Longleftrightarrow  \gamma> \sqrt{\frac{\kappa_1}{\kappa_2}}. 
\eeq
Consequently,
\beqm
X = \left(S-\frac{T}{\gamma^2}\right)^{-1} = \frac{\kappa_2-\kappa_1}{\kappa_2-\frac{\kappa_1}{\gamma^2}} >0
\eeqm
in Eq. \eqref{eq:XY}. As $A_y<0$,  $Y=0$ is a solution to Eq. \eqref{eq:Y}. We have  $\rho(XY) =0$, 
which confirms Eq. \eqref{eq:XY_complex}. By Eq. \eqref{eq:controller_ABC_real},  we have
\beq\label{eq:ABC_K}
\begin{aligned}
&B_K =-\sqrt{\kappa_1}, \quad
C_K  = \frac{\sqrt{\kappa_2}\kappa_1 (1-\gamma^2)}{\gamma^2 \kappa_2-\kappa_1}, 
\\
&A_K =   \frac{-\kappa_2+\kappa_1}{2} - \frac{\kappa_2\kappa_1 (1-\gamma^2)}{\gamma^2 \kappa_2 -\kappa_1}.
\end{aligned}
\eeq

For physical realizability of the quantum controller,  we also need to check
\beq \label{eq:phys_real}
A_K + A_K^\dag + B_K B_K^\dag + C_K^\dag C_K=0.
\eeq
Solving Eq. \eqref{eq:phys_real}
we get
\beqm
\gamma_\pm = \sqrt{\frac{4 \left(\frac{\kappa _1}{\kappa _2}\right)^2\pm\frac{\kappa _1}{\kappa_2} \left(1-\frac{\kappa _1}{\kappa _2}\right) \sqrt{2 \left(1-\frac{\kappa _1}{\kappa _2}\right)}}{\left(\frac{\kappa _1}{\kappa _2}\right)^2+\frac{4 \kappa _1}{\kappa _2}-1}}.
\eeqm
We have the following three cases.
\bei

\item Case 1. When  $\kappa\equiv \frac{\kappa_1}{\kappa_2}>\sqrt{5}-2$, the solution is
\beqm
\gamma =  \sqrt{\frac{4 \kappa^2+ \kappa  (1-\kappa) \sqrt{2 (1- \kappa)}}{\kappa^2+4\kappa-1}} >\sqrt{\kappa}.
\eeqm

\item  Case 2. When $\kappa\equiv\frac{\kappa_1}{\kappa_2}<\sqrt{5}-2$,  the solution is
\beqm \label{eq:case2_dec16}
\gamma =  \sqrt{\frac{4 \kappa^2- \kappa  (1-\kappa) \sqrt{2 (1- \kappa)}}{\kappa^2+4\kappa-1}}  >\sqrt{\kappa}.
\eeqm

\item   Case 3. When $\kappa\equiv\frac{\kappa_1}{\kappa_2}=\sqrt{5}-2$,  solving Eq. \eqref{eq:phys_real} directly we get 
\beqm
\gamma = \frac{1}{2} \sqrt{\frac{13 \sqrt{5}-29}{9-4 \sqrt{5}}}\approx 0.555893> \sqrt{\kappa}.
\eeqm
\eei 

In summary, the physical realizability conditions are always satisfied, and thus no additional quantum vacuum noise needs to be added.   The cost is that we have a higher disturbance attention level $\gamma$ than that given in Eq. \eqref{eq:gamma_cavity}.  This reveals the trade-off between physical realizability and $H^\infty$ disturbance attenuation performance for quantum linear systems.

\subsection{Example 2: DPA}\label{sec:dpa}


A simple example of a non-passive quantum linear system is a DPA, commonly used in quantum optics. A model of a DPA in the real-quadrature operator representation is,   \cite[Eq. (16)]{BZL12}, \cite{SYMF16}, \cite[Eq. (61)]{LDP+22}, \cite[Eqs. (3)-(4)]{LDP+22b}
\begin{equation}\label{eq:DPA_Sept6}
	\begin{aligned}
		\dot{\mbf{x}}=&-\frac{1}{2}\left[
		\begin{array}{cc}
			\kappa_w+\kappa_u-\epsilon  & 0 \\
			0 & \kappa_w+\kappa_u+\epsilon
		\end{array}
		\right] \mbf{x} \\
		&-\frac1{\gamma}\sqrt{\kappa_w}\ \mbf{w}_1-\sqrt{\kappa_u}\ \mbf{u}, \\
		\mbf{z} =&\; \sqrt{\kappa_u} x + \mbf{u},\\
		\mbf{y}=& \sqrt{\kappa_w}\ \mbf{x} +\frac1{\gamma} \mbf{w}_1.
\end{aligned}\end{equation}
For this system,  $\Omega_-=0$, $\Omega_+=\frac{\imath \epsilon}{2}$, $C_-=\sqrt{\kappa_w}$, and $C_+=0$. The parameter $\epsilon$ in $\Omega_+$ designates the strength of the pump field on the DPA.   It is assumed that $\epsilon<\kappa_w+\kappa_u$ to ensure the stability of the DPA. 
By Eq. \eqref{eq:A_x_A_y_dec13},	 
\beqm
\begin{aligned}
	\mathbbm{A}_x =&\; \left[
	\bey{cc}
	\frac{\epsilon-\kappa_w+\kappa_u}{2} & 0\\
	0 & -\frac{\epsilon+\kappa_w-\kappa_u}{2}
	\eey
	\right]
	=  \left[
	\bey{cc}
	\mathbbm{A}_{x1} & 0 \\
	0& \mathbbm{A}_{x3}
	\eey
	\right],\\
	\mathbbm{A}_y =&\;   \left[
	\bey{cc}
	\frac{\epsilon+\kappa_w-\kappa_u}{2} & 0\\
	0 & -\frac{\epsilon-\kappa_w+\kappa_u}{2}
	\eey
	\right]	 = \left[
	\bey{cc}
	\mathbbm{A}_{y1} & 0 \\
	0& \mathbbm{A}_{y3}
	\eey
	\right],
	\end{aligned}
\eeqm
which are consistent with Eq. \eqref{eq:oct27_A_x_y}.  As $\mathbbm{A}_{x2}=0$,  Eq. \eqref{eq:dec17:12} vanishes. Moreover, as $\mathbbm{A}_x$ and $A_y$ are already diagonal, $\mathbbm{W}=I$ in  Eq. \eqref{eq:tilde_A}. As a result, $\mathbbm{Z}=I$ in Eq. \eqref{eq:Z_dec23}. Corollary \ref{cor:iff} applies.

We assume $\kappa_w<\kappa_u$. In this case, $\mathbbm{A}_{x1}>0$ and hence $\mathbbm{A}_{y3}<0$.  We also assume that $\kappa_u\neq \epsilon +\kappa_w$; otherwise  $\mathbbm{A}_{x3}=0$ and assumptions \textbf{A3} and \textbf{A4} do not hold. In what follows, we study two cases.

\subsubsection{The case $\kappa_u>\epsilon +\kappa_w$ } \label{subsec:case 1}

In this case, $\mathbbm{A}_{x3}>0$ and $\mathbbm{A}_{y1}<0$.  Consequently the matrix $\mathbbm{A}_x$ is anti-stable and $\mathbbm{A}_y$ is stable.  Thus, the trivial solution $\mathbbm{Y}=0$ is a stabilizing solution  to Eq. \eqref{eq:Y_0}. Next, we derive $\mathbbm{X}$ that solves  Eq. \eqref{eq:X_0}.  Notice that Eq. \eqref{eq:X_0} is equivalent to
\beqm
-\mathbbm{A}_x^\top \mathbbm{X} - \mathbbm{X} \mathbbm{A}_x - \mathbbm{X}(\frac1{\gamma^2}  \mathbbm{B}_1 \mathbbm{B}_1^\top -\mathbbm{B}_2 \mathbbm{B}_2^\top)\mathbbm{X}=0,
\eeqm
whose solution is 
\beqm
\mathbbm{X} =\frac{\kappa_u-\kappa_w}{\gamma^2} \left[
\bey{cc}
\frac1{\kappa_u+\epsilon-\kappa_w} & 0 \\
0 & \frac1{\kappa_u-\epsilon-\kappa_w} 
\eey
\right].
\eeqm
According to Eq. \eqref{eq:controller_ABC_real}, the controller parameters   are given in Eq. \eqref{eq:controller_ABC_real_DPA_case_1}.
\begin{figure*}
\beq\label{eq:controller_ABC_real_DPA_case_1}
	\begin{aligned}
	\mathbbm{A}_K =&\; \mathbbm{A}+\mathbbm{B}_2\mathbbm{C}_K -\mathbbm{B}_K \mathbbm{C}_2  \\
	=&\;  \left[
	\begin{array}{cc}
		\frac{\gamma^2(\kappa_u^2+\epsilon^2-\kappa_w^2+2\kappa_u\epsilon)+ 2\kappa_u(\kappa_w-\kappa_u)}{2 \gamma ^2 \left(\kappa _u-\kappa _w+\epsilon \right)} & 0 \\
		0 & -\frac{\gamma^2(\kappa_u^2+\epsilon^2-\kappa_w^2-2\kappa_u\epsilon)+ 2\kappa_u(\kappa_w-\kappa_u)}{2 \gamma ^2 \left(\kappa _w-\kappa _u+\epsilon \right)} \\
	\end{array}
	\right], \\
	\mathbbm{B}_K  =&\; \mathbbm{B}_1 = -\sqrt{\kappa_w}I_2,\\
	\mathbbm{C}_K =&\;\sqrt{\kappa _u}  \left[
	\begin{array}{cc}
		 \left(\frac{\kappa _u-\kappa _w}{\gamma ^2 \left(\kappa _u-\kappa _w+\epsilon \right)}-1\right) & 0 \\
		0 & \left(\frac{\kappa _w-\kappa _u}{\gamma ^2 \left(\kappa _w-\kappa _u+\epsilon \right)}-1\right) \\
	\end{array}
	\right],\\
	\tilde{B}_K =&\;  C_K^\sharp, \quad \tilde{C}_K = -B_K^\sharp.
	\end{aligned}
\eeq
\end{figure*}
The physical realizability condition \eqref{eq:PR_real} turns out to be
\begin{align}
	&\gamma ^4 \left[-2 \kappa _u^3+2 \kappa _u^2 \kappa _w+2 \kappa _u \kappa _w^2+2 \epsilon ^2 \kappa _u+2 \kappa _w \left(\epsilon ^2-\kappa _w^2\right)\right]
	\nonumber\\
	&+\gamma ^2 \left(4 \kappa _u^3-8 \kappa _u^2 \kappa _w+4 \kappa _u \kappa _w^2\right)-\kappa _u^3-\kappa _u \kappa _w^2+2 \kappa _u^2 \kappa _w
	\nonumber\\
	=&\; 0.
\end{align}
Solving for $\gamma$ we get two positive solutions, as shown in Eq. \eqref{eq:apr10_temp2}.
\begin{figure*}
\beq \label{eq:apr10_temp2}
\gamma_\pm=\frac{ 2 \kappa _u \left(\kappa _u-\kappa _w\right)^2\pm (\kappa_u-\kappa_w)\sqrt{2\kappa _u  \left[\kappa _u^3-3 \kappa _u^2 \kappa _w+\kappa _u \left(3 \kappa _w^2+\epsilon ^2\right)-\kappa _w^3+\epsilon ^2 \kappa _w\right]}}{2 \left(\kappa _u+\kappa _w\right) \left(\kappa _u-\kappa _w+\epsilon \right) \left(\kappa _u-\kappa _w-\epsilon \right)}.
\eeq
\end{figure*}
Clearly, $\gamma_+\geq \gamma_-$.   Therefore, the controller with parameters in Eq. \eqref{eq:controller_ABC_real_DPA_case_1} is  physically realizable and the minimal closed-loop disturbance attenuation level is  $\gamma_-$. 
%

\subsubsection{The case $\kappa_u<\epsilon +\kappa_w$}

In this case, $\mathbbm{A}_{x3}<0$, and $\mathbbm{A}_{y1}>0$. Setting $\mathbbm{W} = \left[ 
\bey{cc}
0 &1\\
1  &0
\eey
\right]$. Thus $\mathbbm{Z} =I$ in Eq. \eqref{eq:Z_dec23}. Accordingly, Eq. \eqref{eq:XY_dec23_temp2} becomes
\beq \label{eq:XY_dec23_temp2b}
\rho(\mathbbm{X}\mathbbm{Y} ) <1   \Longleftrightarrow 
\rho\left(  \left(\mathbbm{S}-
\frac{\mathbbm{T}}{\gamma^2}
\right)^{-1}    \left(\mathbbm{U}-
\frac{\mathbbm{V}}{\gamma^2}
\right)^{-1}
\right)<\gamma^2.
\eeq
Following the development in  Section \ref{sec:real}, we have 
\beq \label{eq:STUV}
\begin{aligned}
\mathbbm{S} = &\; \frac{\kappa_u}{\epsilon+\kappa_u-\kappa_w} ,\quad \mathbbm{T} =\frac{\kappa_w}{\epsilon+ \kappa_u-\kappa_w}, \\
\mathbbm{U} =&\;  \frac{\kappa_w}{\epsilon+\kappa_w-\kappa_u},\quad \mathbbm{V} =\frac{\kappa_u}{\epsilon+\kappa_w-\kappa_u}
\end{aligned}
\eeq
Noticing
\beqm
\mathbbm{S} -\frac{\mathbbm{T}}{\gamma^2} >0 \Longleftrightarrow   \gamma>\sqrt{ \frac{\kappa_w}{\kappa_u}},
\eeqm
\beqm
\mathbbm{U}-\frac{\mathbbm{V}}{\gamma^2} >0 \Longleftrightarrow   \gamma>\sqrt{ \frac{\kappa_u}{\kappa_w}},
\eeqm
we have 
\beqm
\gamma>\max \left\{\sqrt{ \frac{\kappa_w}{\kappa_u}}, \sqrt{ \frac{\kappa_u}{\kappa_w}}\right\}>1.
\eeqm
Therefore,
\beqm\label{eq:XY_DPA_case_2}
\mathbbm{X}  = \left[
\bey{cc}
\left(\mathbbm{S}-
\frac{\mathbbm{T}}{\gamma^2}
\right)^{-1} & 0 \\
0 & 0
\eey
\right], \ \  
\mathbbm{Y}  =   \frac1{\gamma^2}\left[
\bey{cc}
\left(\mathbbm{U}-
\frac{\mathbbm{V}}{\gamma^2}
\right)^{-1} & 0 \\
0 &  0
\eey
\right]
\eeqm
solve Eq. \eqref{eq:ARE_0}.  Substituting matrices $\mathbbm{S},\mathbbm{T},\mathbbm{U},\mathbbm{V}$ in Eq. \eqref{eq:STUV} into  Eq. \eqref{eq:XY_dec23_temp2b} yields
\beqnm
&&\rho \left(  
\left(\mathbbm{S}-\frac{\mathbbm{T}}{\gamma^2}\right)^{-1} \left(\mathbbm{U}-\frac{\mathbbm{V}}{\gamma^2}\right)^{-1}
\right ) <1 \\
&\Leftrightarrow&
\gamma> \frac1{\sqrt{\frac1{2} (\frac{\kappa_u}{\kappa_w} + \frac{\kappa_w} {\kappa_u})-\frac1{2}\sqrt{(\frac{\kappa_u}{\kappa_w}-\frac{\kappa_w}{\kappa_u})^2+4[\frac{\epsilon^2-(\kappa_w-\kappa_u)^2}{\kappa_u \kappa_w}] }}}.
\eeqnm

After getting $\mathbbm{X}$ and $\mathbbm{Y}$,  the quantum controller $K$ can be constructed with parameters in Eq. \eqref{eq:controller_ABC_real}. 
%
The physical realizability condition \eqref{eq:PR_real} turns out to be Eq. \eqref{eq:feb20_dpa_2}.
\begin{figure*}
\begin{align}
	&\gamma ^6 \left(2 \kappa _u^3 \kappa _w+2 \kappa _u^2 \kappa _w^2-4 \epsilon  \kappa _u^2 \kappa _w\right)+\gamma ^4 \left(-\kappa _u \kappa _w^3-8 \kappa _u^2 \kappa _w^2-3 \kappa _u^3 \kappa _w-\epsilon ^2 \kappa _u \kappa _w+4 \epsilon  \kappa _u \kappa _w^2+4 \epsilon  \kappa _u^2 \kappa _w+2 \epsilon ^3 \kappa _u\right) \nonumber\\
	&\;+\gamma^2  \left(3 \kappa _u \kappa _w^3+8 \kappa _u^2 \kappa _w^2+\kappa _u^3 \kappa _w+\epsilon ^2 \kappa _u \kappa _w-2 \epsilon  \kappa _u \kappa _w^2-\epsilon  \kappa _u^2 \kappa _w-\epsilon ^3 \kappa _w-\epsilon  \kappa _w^3\right)-2 \kappa _u \kappa _w^3-2 \kappa _u^2 \kappa _w^2 \nonumber\\
	=&\; 0. \label{eq:feb20_dpa_2}
\end{align}
\end{figure*}
Solving Eq. \eqref{eq:feb20_dpa_2} indeed gives us an analytical expression of  $\gamma^2$, which is unfortunately too complicated to enable us to check if it is indeed positive given that $\kappa_u<\epsilon +\kappa_w$, without giving numerical values to the parameters. If Eq. \eqref{eq:feb20_dpa_2} indeed has a positive solution $\gamma_+$ satisfying $\gamma_+\geq \max \left\{\sqrt{ \frac{\kappa_w}{\kappa_u}}, \sqrt{ \frac{\kappa_u}{\kappa_w}}\right\}$, then the above controller $K$ is indeed physically realizable  and  then the attenuation level is  $\gamma_+$. Otherwise, additional vacuum inputs should be added to make the controller $K$ quantum-mechanical.

\section{Conclusions}\label{sec:conclu}
In this paper, we investigated the coherent feedback $H^\infty$ control problem for linear quantum systems. Our study reveals that the distinctive properties of this type of quantum systems allow the controller synthesis to be reformulated; specifically, solving at most four Lyapunov equations is sufficient for quantum controller construction, in contrast to the standard requirement of solving two coupled algebraic Riccati equations. The presented results facilitate subsequent studies of coherent feedback $H^\infty$ control in various quantum linear systems, such as those found in quantum optics, superconducting circuits and optomechanics.

\bibliographystyle{ieeetr}
\bibliography{gzhang}

@book{BR04,
  title={A Guide to Experiments in Quantum Optics},
  author={Bachor, Hans-Albert and Ralph, Timothy C},
  year={2004},
  publisher={Wiley}
}

@book{GZ00,
 title={Quantum Noise},
  author={Gardiner, Crispin and Zoller, Peter},
  year={2004},
  publisher={Springer}
}

@article{GJ09,
  title={The series product and its application to quantum feedforward and feedback networks},
  author={Gough, John and James, Matthew R},
  journal={IEEE Transactions on Automatic Control},
  volume={54},
  number={11},
  pages={2530--2544},
  year={2009},
  publisher={IEEE}
}

@article{GJ09b,
  title={Quantum feedback networks: Hamiltonian formulation},
  author={Gough, John and James, Matthew R},
  journal={Communications in Mathematical Physics},
  volume={287},
  number={3},
  pages={1109--1132},
  year={2009},
  publisher={Springer}
}

@article{GJN10,
  title={Squeezing components in linear quantum feedback networks},
  author={Gough, John Edward and James, M. R. and Nurdin, H. I.},
  journal={Physical Review A},
  volume={81},
  number={2},
  pages={023804},
  year={2010},
  publisher={APS}
}

@article{GZ15,
	author={J. E. Gough and G. Zhang},
	title={On realization theory of quantum linear systems},
    journal={Automatica},
	volume={59},
    pages={139-151},
    year={2015},
}

@book{WM08,
  title={Quantum {O}ptics},
  author={Walls, Daniel F and Milburn, Gerard J},
  year={2007},
  publisher={Springer Science \& Business Media}
}

@article{ZJ11,
author={G. Zhang and M. R. James},
title={Direct and indirect couplings in coherent feedback control of linear quantum systems},
    journal={IEEE Transactions on Automatic Control},
volume={56},
    pages={1535-1550},
    year={2011},
}

@article{ZJ12,
  title={Quantum feedback networks and control: a brief survey},
  author={Zhang, GuoFeng and James, Matthew R},
  journal={Chinese Science Bulletin},
  volume={57},
  number={18},
  pages={2200--2214},
  year={2012},
  publisher={Springer}
}

@article{ZLH+12,
  title={Coherent feedback control of linear quantum optical systems via squeezing and phase shift},
  author={Zhang, Guofeng and Joseph Lee, Heung Wing and Huang, Bo and Zhang, Hu},
  journal={SIAM Journal on Control and Optimization},
  volume={50},
  number={4},
  pages={2130--2150},
  year={2012},
  publisher={SIAM}
}

@article{ZJ13,
  title={On the response of quantum linear systems to single photon input fields},
  author={Zhang, Guofeng and James, Matthew R},
  journal={IEEE Transactions on Automatic Control},
  volume={58},
  number={5},
  pages={1221--1235},
  year={2013},
  publisher={IEEE}
}

@article{ZGPG18,
  title={The {K}alman decomposition for linear quantum systems},
  author={Zhang, Guofeng and Grivopoulos, Symeon and Petersen, Ian R and Gough, John E},
  journal={IEEE Transactions on Automatic Control},
  volume={63},
  number={2},
  pages={331--346},
  year={2018},
  publisher={IEEE}
}

@article{ZLW+17,
  title={Quantum feedback: theory, experiments, and applications},
  author={Zhang, Jing and Liu, Yu-xi and Wu, Re-Bing and Jacobs, Kurt and Nori, Franco},
  journal={Physics Reports},
  volume={679},
  pages={1--60},
  year={2017},
  publisher={Elsevier}
}

@article{JNP08,
  title={${H}^\infty$ control of linear quantum stochastic systems},
  author={James, Matthew R and Nurdin, Hendra I and Petersen, Ian R},
  journal={IEEE Transactions on Automatic Control},
  volume={53},
  number={8},
  pages={1787--1803},
  year={2008},
  publisher={IEEE}
}

@article{NJD09,
  title={Network synthesis of linear dynamical quantum stochastic systems},
  author={Nurdin, Hendra I and James, Matthew R and Doherty, Andrew C},
  journal={SIAM Journal on Control and Optimization},
  volume={48},
  number={4},
  pages={2686--2718},
  year={2009},
  publisher={SIAM}
}

@article{HP84,
  title={Quantum {Ito's} formula and stochastic evolutions},
  author={Hudson, Robin L and Parthasarathy, Kalyanapuram R},
  journal={Communications in Mathematical Physics},
  volume={93},
  number={3},
  pages={301--323},
  year={1984},
  publisher={Springer}
}

@book{NY17,
  title={Linear Dynamical Quantum Systems - Analysis, Synthesis, and Control},
  author={Nurdin, Hendra I and Yamamoto, Naoki},
  year={2017},
  publisher={Springer-Verlag Berlin}
}

@article{NJP09,
  title={Coherent quantum {LQG} control},
  author={Nurdin, Hendra I and James, Matthew R and Petersen, Ian R},
  journal={Automatica},
  volume={45},
  number={8},
  pages={1837--1846},
  year={2009},
  publisher={Elsevier}
}

@book{ZDG96,
  title={Robust and Optimal Control},
  author={Zhou, Kemin and Doyle, J. C. and Glover, K.},
  year={1996},
  publisher={Prentice Hall, Englewood Cliffs, New Jersey}
}

@article{Mabuchi08,
  title={Coherent-feedback quantum control with a dynamic compensator},
  author={Mabuchi, Hideo},
  journal={Physical Review A},
  volume={78},
  number={3},
  pages={032323},
  year={2008},
  publisher={APS}
}

@book{WM10,
  title={Quantum Measurement and Control},
  author={Wiseman, H. M. and Milburn, G. J.},
  year={2010},
  publisher={Cambridge University Press}
}

@article{SY21,
  title={Quantum functionalities via feedback amplification},
  author={Shimazu, Rion and Yamamoto, Naoki},
  journal={Physical Review Applied},
  volume={15},
  number={4},
  pages={044006},
  year={2021},
  publisher={APS}
}

@ARTICLE{PDP17,
  title={Dark modes of quantum linear systems},
  author={Pan, Y. and Dong, D. and Petersen, I.R.},
  journal={IEEE Transactions on Automatic Control},
  volume={62},
  issue={8},
  pages={4180--4186},
  year={2017},
}

@ARTICLE{petersen11,
  title={Cascade cavity realization for a class of complex transfer functions arising in coherent quantum feedback control},
  author={Petersen, I. R.},
  journal={Automatica},
  volume={47},
  issue={8},
  pages={1757--1763},
  year={2011},
}

@ARTICLE{XPD17,
  title={Coherent robust {$H^\infty$} control of linear quantum systems with uncertainties in the {H}amiltonian and coupling operators},
  author={Xiang, C. and Petersen, I. R. and Dong, D.},
  journal={Automatica},
  volume={81},
  pages={8--21},
  year={2017},
}

@article{ZD22,
title = {Linear quantum systems: a tutorial},
journal = {Annual Reviews in Control},
volume = {54},
pages = {274-294},
year = {2022},
issn = {1367-5788},
author = {Guofeng Zhang and Zhiyuan Dong},
}

@book{GvL13,
  title={Matrix Computations, 4th ed.},
  author={Golub, G and Van Loan, C},
  year={2013},
  publisher={Johns Hopkins University Press, Baltimore}
}

@article{Nurdin2010b,
  title={On synthesis of linear quantum stochastic systems by pure cascading},
  author={Nurdin, Hendra Ishwara},
  journal={IEEE Transactions on Automatic Control},
  volume={55},
  number={10},
  pages={2439--2444},
  year={2010},
  publisher={IEEE}
}

@article{BZL12,
author = {Chuanxin Bian and Guofeng Zhang and Heung Wing Joseph Lee},
title = {Squeezing enhancement of degenerate parametric amplifier via coherent feedback control},
journal = {International Journal of Control},
volume = {85},
number = {12},
pages = {1865--1875},
year = {2012},
publisher = {Taylor \& Francis},
}

@article{BQD24,
title = {Stabilizing preparation of quantum Gaussian states via continuous measurement},
journal = {Automatica},
volume = {164},
pages = {111622},
year = {2024},
issn = {0005-1098},
doi = {https://doi.org/10.1016/j.automatica.2024.111622},
author = {Liying Bao and Bo Qi and Daoyi Dong},
}

@article{BQND24,
  title = {Exponential sensitivity revival of noisy non-{H}ermitian quantum sensing with two-photon drives},
  author = {Bao, Liying and Qi, Bo and Nori, Franco and Dong, Daoyi},
  journal = {Physical Review Research},
  volume = {6},
  issue = {2},
  pages = {023216},
  numpages = {18},
  year = {2024},
  month = {May},
  publisher = {American Physical Society},
  doi = {10.1103/PhysRevResearch.6.023216},
}

@article{PAJ91,
	author = {Petersen, Ian R. and Anderson, Brian D. O. and Jonckheere, Edmond A.},
	title = {A first principles solution to the non-singular {$H^\infty$} control problem},
	journal = {International Journal of Robust and Nonlinear Control},
	volume = {1},
	number = {3},
	pages = {171-185},
	doi = {https://doi.org/10.1002/rnc.4590010304},
	url = {https://onlinelibrary.wiley.com/doi/abs/10.1002/rnc.4590010304},
	eprint = {https://onlinelibrary.wiley.com/doi/pdf/10.1002/rnc.4590010304},
	year = {1991}
}

@ARTICLE{P89,
	author={Petersen, Ian R.},
	journal={IEEE Transactions on Automatic Control}, 
	title={Complete results for a class of state feedback disturbance attenuation problems}, 
	year={1989},
	volume={34},
	number={11},
	pages={1196-1199},
	doi={10.1109/9.40752}}

@article{leonhardt2003,
	title={Explicit effective Hamiltonians for general linear quantum-optical networks},
	author={Leonhardt, Ulf and Neumaier, Arnold},
	journal={Journal of Optics B: Quantum and Semiclassical Optics},
	volume={6},
	number={1},
	pages={L1},
	year={2003},
	publisher={IOP Publishing}
}

@article{DZLP26,
		author={Dong, Zhiyuan and Zhang, Guofeng and Lee, Heung-wing Joseph and Petersen, Ian R.},
		journal={IEEE Transactions on Automatic Control}, 
		title={Linear quantum systems: poles, zeros, invertibility and sensitivity}, 
		year={2026},
		volume={},
		number={},
		pages={1-16},
		doi={10.1109/TAC.2026.3680121}}

@article{MP11b,
	author = {Maalouf, Aline and Petersen, Ian},
	year = {2011},
	month = {03},
	pages = {309 - 319},
	title = {Coherent {$H^{\infty }$} Control for a Class of Annihilation Operator Linear Quantum Systems},
	volume = {56},
	journal = {IEEE Transactions on Automatic Control},
	doi = {10.1109/TAC.2010.2052942}
}

@article{LDP+22,
	title={Fault-Tolerant Coherent {$H^\infty$} Control for Linear Quantum Systems},
	author={Liu, Yanan and Dong, Daoyi and Petersen, Ian R and Gao, Qing and Ding, Steven X and Yokoyama, Shota and Yonezawa, Hidehiro},
	journal={IEEE Transactions on Automatic Control},
	volume={67},
	number={10},
	pages={5087--5101},
	year={2022},
	publisher={IEEE}
}

@article{LDP+22b,
	title={Fault-tolerant {$H^\infty$} control for optical parametric oscillators with pumping fluctuations},
	author={Liu, Yanan and Dong, Daoyi and Petersen, Ian R and Yonezawa, Hidehiro},
	journal={Automatica},
	volume={140},
	pages={110236},
	year={2022},
	publisher={Elsevier}
}

@ARTICLE{P13,
	author={Petersen, Ian R.},
	journal={IEEE Transactions on Automatic Control}, 
	title={Singular Perturbation Approximations for a Class of Linear Quantum Systems}, 
	year={2013},
	volume={58},
	number={1},
	pages={193-198},
	doi={10.1109/TAC.2012.2203030}}

@article{CDZL17,
	title={Mixed {LQG} and {$H^\infty$} coherent feedback control for linear quantum systems},
	author={Cui, Lei and Dong, Zhiyuan and Zhang, Guofeng and Lee, Heung Wing Joseph},
	journal={International Journal of Control},
	volume={90},
	number={12},
	pages={2575--2588},
	year={2017},
	publisher={Taylor \& Francis}
}

@article{MP12,
	title={Time-varying {$H^\infty$} control for a class of linear quantum systems: a dynamic game approach},
	author={Maalouf, Aline I and Petersen, Ian R},
	journal={Automatica},
	volume={48},
	number={11},
	pages={2908--2916},
	year={2012},
	publisher={Elsevier}
}

@article{nurdin2010synthesis,
	title={Synthesis of linear quantum stochastic systems via quantum feedback networks},
	author={Nurdin, Hendra I},
	journal={IEEE Transactions on Automatic Control},
	volume={55},
	number={4},
	pages={1008--1013},
	year={2010},
	publisher={IEEE}
}

@article{SYMF16,
	title={Creation and measurement of broadband squeezed vacuum from a ring optical parametric oscillator},
	author={Serikawa, Takahiro and Yoshikawa, Jun-ichi and Makino, Kenzo and Furusawa, Akira},
	journal={Optics Express},
	volume={24},
	number={25},
	pages={28383--28391},
	year={2016},
	publisher={Optical Society of America}
}

@article{VP16,
	title={Quantum noises, physical realizability and coherent quantum feedback control},
	author={Vuglar, Shanon L and Petersen, Ian R},
	journal={IEEE Transactions on Automatic Control},
	volume={62},
	number={2},
	pages={998--1003},
	year={2017},
	publisher={IEEE}
}

@ARTICLE{YK03b,
	author={Yanagisawa, M. and Kimura, H.},
	journal={IEEE Transactions on Automatic Control}, 
	title={Transfer function approach to quantum Control-Part II: Control concepts and applications}, 
	year={2003},
	volume={48},
	number={12},
	pages={2121-2132},
	doi={10.1109/TAC.2003.820065}}

@article{UJ24,
	title = {Design of coherent passive quantum equalizers using robust control theory},
	journal = {Automatica},
	volume = {163},
	pages = {111599},
	year = {2024},
	issn = {0005-1098},
	doi = {https://doi.org/10.1016/j.automatica.2024.111599},
	url = {https://www.sciencedirect.com/science/article/pii/S000510982400092X},
	author = {Valery Ugrinovskii and Matthew R. James}
}

@article{XLDPU26,
	title={Two-disk synthesis of coherent quantum passive equalizers for channels with Hamiltonian uncertainty},
	author={Xiao, Shuixin and Liang, Weichao and Dong, Daoyi and Petersen, Ian R and Ugrinovskii, Valery},
	journal={IEEE Transactions on Automatic Control},
	year={2026},
	publisher={IEEE}
}

@book{DP23,
	title={Learning and {R}obust {C}ontrol in {Q}uantum {T}echnology},
	author={Dong, Daoyi and Petersen, Ian R},
	year={2023},
	publisher={Springer}
}

@book{BK95,
	title={Quantum {M}easurement},
	author={Braginsky, Vladimir B and Khalili, Farid Ya},
	year={1995},
	publisher={Cambridge University Press}
}

\end{document}